\documentclass[12pt]{article}
\usepackage{hyperref}
\usepackage[round]{natbib}
\usepackage{authblk}
  \usepackage{pdfpages}
  \usepackage{amsmath,amssymb,amsfonts,mathrsfs,amsthm}

\usepackage{bbm}

\usepackage{epsf}
\usepackage{epsfig}
\usepackage{setspace}
\usepackage{multirow}
\usepackage{subfigure}
\usepackage{comment}

\usepackage{graphics}
\usepackage{graphicx}
\usepackage{geometry}
\usepackage{array}
\usepackage{float}
\usepackage{lscape}
\usepackage{color}
\usepackage{calligra}
\usepackage{dsfont}
\usepackage{tabularx}
\usepackage{dcolumn}
\usepackage[bottom]{footmisc}
\usepackage{enumerate}
\usepackage{url}
\usepackage[normalem]{ulem}
\usepackage{epstopdf}
\usepackage[space]{grffile}
\usepackage{setspace}

\usepackage{accents}

\newtheorem{theorem}{Theorem}

\newtheorem{corollary}{Corollary}

\newtheorem{example}{Example}
\newtheorem{remark}{Remark}

\usepackage{mathtools}

\DeclareMathAlphabet{\mathpzc}{OT1}{pzc}{m}{it}
\DeclareMathAlphabet{\mathcalligra}{T1}{calligra}{m}{n}
\title{Axioms for Constant Function Market Makers}
\author[1]{Jan Christoph Schlegel}
\author[2]{Mateusz Kwaśnicki}
\author[3]{Akaki Mamageishvili}

\affil[1]{Department of Economics, City, University of London, London, UK\footnote{Corresponding author: \small\tt jansc@alumni.ethz.ch}}
\affil[2]{Department of Pure Mathematics, Wroclaw University of Science and Technology, Wroclaw, Poland}
\affil[3]{Offchain Labs, Zurich, Switzerland}

\date{}

\begin{document}


\maketitle


\begin{abstract}
    We study axiomatic foundations for different classes of constant-function automated market makers (CFMMs). We focus particularly on separability and on different invariance properties under scaling. Our main results are an axiomatic characterization of a natural generalization of constant product market makers (CPMMs), popular in decentralized finance, on the one hand, and a characterization of the Logarithmic Scoring Rule Market Makers (LMSR), popular in prediction markets, on the other hand. The first class is characterized by the combination of independence and scale invariance, whereas the second is characterized by the combination of independence and translation invariance. The two classes are therefore distinguished by a different invariance property that is motivated by different interpretations of the numéraire in the two applications. 
    However, both are pinned down by the same separability property.
    Moreover, we characterize the CPMM as an extremal point within the class of scale invariant, independent, symmetric AMMs with non-concentrated liquidity provision. Our results add to a formal analysis of mechanisms that are currently used for decentralized exchanges and connect the most popular class of DeFi AMMs to the most popular class of prediction market AMMs.  
\end{abstract}

\section{Introduction}
One of the first and so far most successful applications of Decentralized Finance (DeFi), financial applications run on permissionless blockchains, are so-called Automated Market Makers (AMMs). They are used to trade cryptocurrencies algorithmically without relying on a custodian or trusted third party. The average daily volume traded on AMMs used in DeFi reached \$3.5B in 2022, while the total value locked in DeFi protocols at the time of writing this paper\footnote{January 20th, 2023, data obtained from defilama.com.} is estimated to be \$45B. The considerable volume of assets traded on AMMs contrasts with their simple economic design which is notably different from the design of traditional financial exchanges. The state of a typical AMM used in DeFi consists of the current inventories of the traded tokens. Trades are made such that some invariant of these inventory sizes is kept constant. Traders who want to exchange tokens of type $A$ for tokens of another type $B$, add $A$ tokens to the inventory and in return obtain an amount of $B$ tokens from the inventory so that the invariant is maintained. The simplicity of these so-called constant function market makers (CFMMs) stems from the limited storage and computation capabilities of smart contract blockchains. In comparison to more conventional market designs such as limit order books, updating a simple function and inventories is relatively inexpensive to implement and scalable on smart contract blockchains such as Ethereum. Moreover, AMMs allow market participants to provide liquidity passively without having to trade themselves. This makes liquidity provision in AMMs a popular way of getting (negative) exposure to volatility while being compensated by a steady stream of fee income.

While CFMMs proved to be very popular and reliable, the construction of invariants to define them seems in many ways ad-hoc and not based in much theory. In this paper, we fill this gap and propose an axiomatic approach to constructing CFMMs. The approach is, as in any axiomatic theory, to formalize simple principles that are implicitly or explicitly used when constructing trading functions in practice and to check which classes of functions satisfy these principles, beyond those functions already used in practice. The constant product market maker (CPMM), which makes the market in such a way that the product of inventories of the different traded assets is kept constant, has been particularly focal in DeFi applications. The CPMM was originally introduced by Uniswap, the first and so far most popular DeFi AMM by trading volume.\footnote{A natural generalization of the CPMM, the multi-dimensional weighted geometric mean is for example used by Balancer,~\cite[]{balancer2019}, the third largest DeFi AMM by trading volume. The CPMM was used by Uniswap before its update to Version 3 in 2021 and is still in use in many other AMMs. The CPMM was replaced in V3 by a new design~\cite[]{uniswap} that abandons fungibility of liquidity to allow for customized liquidity positions that can be chosen by individual liquidity providers. Fungibility of liquidity positions is captured by the scale invariance axioms that we discuss subsequently. The second largest DeFi AMM by trading volume, the Curve protocol~\cite[]{egorov2019}, is specifically designed to trade highly correlated assets (such as stable coins or different representations of the same cryptocurrency).  Our second main axiom, independence, is violated by the curve AMM. As we will argue, our main characterization implies separability of the market-making function, which is mainly appealing for situations where the price correlation between the traded tokens is imperfect. However, for the two-dimensional case, Theorem~\ref{thm:bijection} characterizes a broad class (including the two-dimensional version of the curve formula) of generally non-separable CFMMs.}
Our first main result (Corollary~\ref{main_result}) characterizes a natural superclass of the (multi-dimensional) CPMM by three natural axioms. The CPMM rule is characterized by being trader optimal within this class (Theorem~\ref{opt:product}). 

In contrast to the case of AMMs for DeFi, there is a much more developed theory for AMMs for prediction markets (see the literature review, Section~\ref{litreview}), for which AMMs have originally been introduced~\cite[]{hanson2003}. There are notable differences in the way that AMMs for prediction markets and those for DeFi function in practice, which has to do with the different nature of the assets traded in them: In a prediction market, the assets are artificially created Arrow-Debreu securities that allow traders to bet on an outcome. Moreover, traders exchange assets against a numéraire rather than swapping between assets.  However, there are also notable similarities that allow for a common analysis of the two applications, as we discuss below.
Similarly as for the DeFi context, for prediction markets, one rule has become focal in theory and practice. This is Hanson's logarithmic scoring rule market maker (LMSR)~\cite{prediction_markets} which maintains an invariant that is the logarithm of a sum of exponentials of the inventories.\footnote{In its original formulation, the LMSR is defined by a logarithmic scoring rule. AMMs defined by scoring rules can be equivalently described by a cost function~\cite[]{Chen2007} which can be shown to take the form of a sum of exponentials for the logarithmic scoring rule. See Footnote~\ref{footnote} for the equivalence between the two definitions. Describing AMMs in terms of cost of trading and in terms of the value of the inventory is equivalent up to change in sign as we discuss below.} Our second main result characterizes the LMSR by three natural axioms (Corollary~\ref{cor:trans}). To the best of our knowledge, this is the first axiomatic characterization of this rule. Our results give a possible normative justification for the use of the CPMM and of the LMSR rule. Moreover, they highlight a hidden commonality between the two rules; the main axiom driving both characterization results is the same independence axiom.

Independence requires that the terms of trade for trading a subset of assets should not depend on the inventory levels of not-traded assets. In the case of smooth liquidity curves, this is equivalent to requiring that the exchange rates for traded asset pairs do not depend on the inventory levels of assets not involved in the trade. 
If AMMs are defined through cost functions, as is customary in the prediction market context, independence (equivalently) means the following: suppose a trader trades in a subset of assets. Then whether two trades involving these assets cost the same should be independent of previous trades made in other assets. Our characterization results combine independence with one of two invariance properties that derive their normative appeal through the different role the numéraire plays in DeFi and in prediction markets.

In the DeFi application, the numéraire is an "LP token" which is a derivative product that is a claim to a fraction of the pooled liquidity and the accrued trading fees of the AMM. A desirable property of these liquidity positions is that they are fungible. Practically, the fungibility of liquidity positions allows tokenizing them in order to use them in other applications, for example as collateral or to combine them with other assets to new financial products. This is an instance of what is usually called the "composability" property of DeFi protocols. Fungibility of liquidity positions requires scale invariance, or, in a stronger form, homogeneity of the trading function.  Geometrically, scale invariance means that liquidity curves through different liquidity levels can be obtained from each other, by projection along rays through the origin, analogous to how homothetic preferences in consumer theory can be constructed. Scale invariance of the inventory has been assumed or discussed in previous work, usually in the stronger form of (positive) homogeneity~\citep{capponi2021} or (positive) linear homogeneity~\citep{Othman2021,Angeris2022abs} and most popular AMMs used in practice, including the CPMM, satisfy it.

In the prediction market application, the numéraire is external to the AMM (usually it is a regular currency such as US dollars). Traders exchange assets against the numéraire rather than swapping between assets. The counterparty at settlement is the prediction market organizer who sells the securities to traders and therefore is the sole liquidity provider. Assets in a prediction market are Arrow-Debreu securities and the market is sufficiently complete so that traders can hold a risk-less portfolio consisting of one unit of each security so that the portfolio always pays out one unit of numéraire in every state of the world. As a consequence of this, AMMs for a prediction market
usually satisfy a different invariance property: {\it translation invariance}, requiring that making a risk-less trade consisting of the same amount of each of the assets always costs the same independently of the state of the AMM. In the case of a smooth trading function, this equivalently means that marginal prices of the traded assets measured in the numéraire always add up to one (see e.g.~\citet{Chen2007}). Translation invariance of AMMs is usually assumed in the literature on AMMs for prediction markets, see Section~\ref{litreview}, and AMMs used in practice, including the LMSR, satisfy it.

The combination of scale invariance or translation invariance and independence imply that the terms of trade are fully determined by the inventory ratio resp. the inventory difference of the pair traded. Under scale invariance, at the margin, percentage changes in exchange rates are proportional to percentage changes in inventory ratio. Under translation invariance, at the margin, percentage changes in exchange rates are proportional to changes in the difference of the two inventories. Therefore, combining independence and scale invariance, we obtain the class of constant inventory elasticity\footnote{The inventory elasticity of an AMM is a measure of how marginal prices change with the ratio of the inventories of the traded assets. If the inventory elasticity is $\epsilon$ when the AMM holds $x$ units of asset $A$ and $y$ units of asset $B$, then changing  the inventory ratio $y/x$ by $1\%$ corresponds to changing the exchange rate of $A$ and $B$ by $1/\epsilon\%$. For example, for the constant product rule, the inventory elasticity is $1$ since the exchange rate is simply given by the inventory ratio $y/x.$ } market makers (CEMMs) (Theorem~\ref{main_result}). This general class contains as special cases constant products, weighted geometric means as well as weighted means. On the other hand, by combining independence and translation invariance, we obtain LMSR rules or constant sum market makers (Theorem~\ref{secondmain}).
Both classes also contain AMMs with non-convex liquidity curves\footnote{For CEMMs, the non-convex case corresponds to negative inventory elasticity, where increases in inventory ratio correspond to decreases in exchange rates. For LMSRs the non-convex case corresponds to a negative liquidity parameter $b$.} which are usually not appealing in practice. 
 We can, however, combine the two axioms in either characterization with other axioms to obtain  characterization results for convex liquidity curves: If, in addition to scale invariance and independence, we impose that liquidity curves are convex, then the elasticity in the above characterization is positive;  alternatively, if we require un-concentrated liquidity (liquidity curves should not intersect with the axes) then the elasticity in the above characterization is positive but smaller or equal to $1$ (Corollary~\ref{cor:aversion}). Similarly, if, in addition to translation invariance and independence, we impose convexity of liquidity curves,\footnote{Convexity of liquidity curves follows from concavity of the trading function which is equivalent to the convexity of the cost function which is a standard assumption in the prediction market literature.} then the liquidity parameter $b$ in the LMSR formula is non-negative (where the constant sum market maker corresponds to the limit of the LMSR as $b$ approaches infinity, Corollary~\ref{cor:trans}).

 Finally, if we further require symmetry in market making, the AMMs in the class of scale invariant, independent AMMs with un-concentrated liquidity can be ranked by the curvature of their liquidity curves which determines how favorably the terms of trade are from the point of view of traders; the CPMM is characterized by being trader optimal within this class (Theorem~\ref{opt:product}).


The above characterizations are obtained for the case of more than two assets traded in the AMM. For the case of exactly two assets, the independence axiom is trivially satisfied and we generally obtain a much larger class of trading functions satisfying the above axioms (Theorem~\ref{thm:bijection}). The class can no longer be completely ranked by the convexity of the induced liquidity curves. However, if we focus on separable CFMMs we obtain the same kind of characterizations (Theorems~\ref{elasticity} and~\ref{elasticity2} and Corollaries~\ref{cor:3} and~\ref{cor:4}) as in the multi-dimensional case, as well as the same kind of optimality result for the CPMM (Theorem~\ref{opt:product2}). In the two-dimensional case, separability of the trading function is a consequence of an additivity property for liquidity provision that we call Liquidity Additivity.

The axiomatic approach leads us to considerations and classes of functions familiar from other fields in economics, consumer theory, and production theory in particular where constant elasticity functions play an important role.
 A main technical contribution of the present paper is to derive the constant elasticity functional form on the one hand and the LogSumExp form on the other hand from two natural properties, without relying on any differentiability or on convexity assumptions.

The paper is organized in the following way. In the next subsection, we briefly review the literature. In Section~\ref{model_preliminaries} we introduce notation and various axioms. In Section~\ref{multiple}, we provide characterization results for the case of more than two assets. In Section~\ref{results}, we provide results for the case of  two assets. In Appendix~\ref{LogicalIndependence}, we show that the axioms used in the various characterizations are logically independent. In Appendix~\ref{Discontinuous} we show how our result fails to hold for discontinuous liquidity curves.

\subsection{Related Work}\label{litreview}
Automated market makers have first been analyzed scientifically in the context of prediction markets~\cite[]{hanson2003,prediction_markets,AMM_prediction, Chen2007,chen2008,Abernethy2011,Ostrovsky2012,othman2013}. The original formulation of AMMs uses scoring rules~\cite[]{hanson2003}. Traders who interact with the AMM obtain a net pay-off after settlement that is determined by a scoring rule that measures the accuracy of the trader's prediction in comparison to the realized state. In the case of the LMSR the score is logarithmic. \citet{Chen2007} observe that AMMs can alternatively be defined in terms of cost functions. In this paper, we work with trading functions that are equivalent to cost functions (see below).
Translation invariance is generally necessary for the equivalence between scoring rule AMMs and cost function AMMs~\citep{Chen2007}. However, in subsequent work to ours, \cite{Bo} show that the (not translation invariant) CPMM in two dimensions can be described by a family of scoring rules where each scoring rule describes trading along a fixed liquidity curve. The construction generalizes to other CFMMs and scale invariance of the CFMM guarantees that scoring rules for different liquidity levels can be obtained by scaling with a linear factor (for the CPMM the scoring rules are given by $S_A(p)=-k\sqrt{p_A/p_B}$ where $k$ depends on the liquidity level). Thus, non-translation invariant CFMMs can be described by families of scoring rules rather than a single scoring rule.
\cite{othman2011liquidity} observe that prediction market AMMs satisfying translation invariance are also mathematically equivalent to risk measures from finance~\cite[]{delbaen} which gives another representation of them. \cite{Ostrovsky2012} analyzes information aggregation in markets made according to a scoring rule AMM. Information is successfully aggregated if the assets satisfy a separability property.~\cite{Spyros} extends the analysis to ambiguity averse traders.

Similarly, as in the prediction market use case, DeFi AMMs can be used to aggregate information into a price signal. In DeFi, "price oracles"  provide exchange rates between tokens to other blockchain applications, such as lending protocols, eliminating the need to obtain price information from an off-chain source.
In~\citet{CFMM} the authors study CFMMs as a device to create a price oracle.~\citet{composability} study axiomatization of composable AMMs. AMMs allowing exchanges of more than two assets arise naturally by merging different AMMs exchanging only two assets. That is, a trader can split his trade across different AMMs. The problem of optimizing the order flow over different AMMs has been studied in~\citet{ChitraEC}.
Several papers consider trading and liquidity provision behavior for CPMMs.~\citet{Park2021} argues that CPMMs are susceptible to arbitrage opportunities in the presence of multiple competing exchanges. To fix the potential drain of resources from the AMM,~\citet{dynamic_curves} proposes dynamic curves that use an additional input from a market price oracle to dynamically adjust the liquidity curve. \citet{aoyagi2020} determines equilibrium liquidity and expected profits for liquidity providers for CPMMs. \citet{Lehar2021} provide an empirical study of liquidity provision and pricing in a decentralized exchange (Uniswap) in comparison to a centralized exchange (Binance). \citet{capponi2021} study the incentives of liquidity providers and investors in AMMs in the presence of arbitrage. They also link it to design, the curvature of the trading function, and the number of pooled assets. In independent work from ours (the paper was released shortly after our original working paper),~\citet{bichuch2022} study several properties of two-dimensional AMMs, check whether popular examples satisfy these properties, and study fees and impermanent loss.
Our paper also relates to the earlier literature in economics on CES functions in utility and production theory. In particular,~\citet{RADER1981} shows in the context of consumption choices over time that separable, homothetic, and quasi-concave utility functions are of the CES type. Our first theorem proves a mathematically related result,  however, we do not rely on quasi-concavity for our proof.



\section{Model and Axioms}\label{model_preliminaries}
A constant function market maker (CFMM) for a finite set of assets $\mathcal{A}=\{A,B,\ldots\}$ maintains an {\bf inventory} $I\in\mathbb{R}_+^{\mathcal{A}}$ of assets and makes a market according to a {\bf trading function}
 $f:\mathbb{R}^{\mathcal{A}}_+\rightarrow\mathbb{R}$ where $f(I)$ is the value of the inventory measured in a numéraire. We make two assumptions throughout: The trading function is {\bf strictly increasing} in all variables for strictly positive inventories and {\bf continuous} in the sense that for each $J\in\mathbb{R}_+^{\mathcal{A}}$ the sets $\{I:f(I)\geq f(J)\}$ and $\{I:f(I)\leq f(J)\}$ are closed (some of our results only require continuity on the boundary, see Appendix~\ref{Discontinuous}). Traders interact with the AMM by making {\bf trades} $r\in\mathbb{R}^{\mathcal{A}}$ such that $I\geq r$. The new inventory of the AMM after trading is $I-r$ and the trader needs to provide $f(I)-f(I-r)$ units of numéraire to the market maker to execute the trade.

\subsubsection*{Decentralized Exchanges}
In decentralized exchanges, the inventory itself is usually tokenized and the unit of these "liquidity provider (LP) tokens" can be used as numéraire. Moreover, trading is usually restricted to take one of two forms  
\begin{enumerate}
    \item swapping assets without involving the numéraire. Without transaction fees, this corresponds to trades $r\in\mathbb{R}^{\mathcal{A}}$ with $$f(I-r)=f(I).$$ The set of all inventory levels $I-r$ for trades $r$ satisfying the previous equation is called a {\bf liquidity curve} (or surface). 
    \item adding or removing liquidity in exchange for the numéraire. Without transaction fees, this corresponds to trades $r$ that are multiples of the current inventory $$r=\lambda I\quad\text{ for }\lambda\geq-1,$$
    where the trader obtains $f(I-r)-f(I)$ units of the numéraire if he provides liquidity ($\lambda<0$) and pays $f(I)-f(I-r)$ units of the numéraire if he removes liquidity ($\lambda>0$).
\end{enumerate}
 Without transaction fees and assuming that the trader is endowed with enough numéraire good (or can borrow it), the restriction to  two types of trades is without loss of generality: Any trade $r\in\mathbb{R}^{\mathcal{A}}$ such that $I\geq r$ can be realized by a combination of swaps, liquidity addition, and removal. In particular, a trader can exchange any asset$A\in \mathcal{A}$ for the numéraire by first swapping some amount of asset $A$ other assets and afterward providing liquidity.
\subsubsection*{Prediction markets}
In the prediction market literature~\citep{Abernethy2011}, AMMs are often defined in terms of an (increasing, continuous) {\bf cost function} $C:\mathbb{R}^{\mathcal{A}}\rightarrow \mathbb{R}$.
Here $C(r)$ is the cost measured in the numéraire and incurred by traders from executing trade $r\in\mathbb{R}^{\mathcal{A}}$. In contrast to the decentralized exchange application, trading in prediction markets is usually assumed to involve the numéraire whereas liquidity is provided by a single market operator.
However, the cost function formulation is equivalent to the trading function formulation as can be seen as follows: Suppose that the initial inventory of the AMM before trading is $I^0$. Then we can define a cost function by $$C(r)=-f(I^0-r).$$ The cost of executing a trade $r$  is equal to the loss of value of inventory if the inventory changes from $I^0$ to $I^0-r$, which is $f(I^0-r)-f(I^0)$. More generally, for a sequence of trades $r^1,\ldots, r^T$ so far realized, the cost of making a new trade $r$ is $C(\sum_{t=1}^T r^t+r)-C(\sum_{t=0}^Tr^t)$ which is equal to the change in inventory $f(I^T-r)-f(I^T)$ where $I^T=I^0-\sum_{t=1}^Tr^t$.

\begin{remark}
Our model implicitly assumes path independence of the AMM. This refers to the assumption that the trading function is a function only of the current inventory, but not of the past history of trades. See~\citet{Abernethy2011} for a general model of path-dependent AMMs and a formal statement of the path independence axiom that we implicitly assume by focusing on trading functions that are a function of current inventories alone. CFMMs used in practice usually satisfy path independence if we abstract away from transaction fees.
\end{remark}

We introduce the following equivalence notion for CFMMs: two trading functions $f,g$ are {\bf equivalent} (for market making) if for all inventories $I,J$ we have $$f(I)=f(J)\Leftrightarrow g(I)=g(J).$$
Two equivalent CFMMs will always make the market in the same way in the sense that they induce the same family of liquidity curves and therefore offer the same terms of trade at the same inventory levels. However, they can differ in the cardinal measure of liquidity, e.g. by using different numéraires.
Since trading functions $f,g$ are increasing in all arguments, $f$ and $g$ are equivalent if and only if there is a strictly increasing function $M$ on the range of $f$ such that $M(f(I))=g(I)$ for each $I\in\mathbb{R}_+^{\mathcal{A}}$.

Popular examples of trading functions used in practice are:
\begin{itemize}
\item {\bf weighted geometric means}~\cite[]{balancer2019} where $$f_{\text{product}}(I)=\prod_{A\in\mathcal{A}}I^{\alpha_A}_A,$$
which in the symmetric case where $\alpha_A=\alpha_B$ for $A,B\in\mathcal{A}$ is the popular {\bf constant product market maker (CPMM)},
\item {\bf constant sum market makers} $$f_{\text{mean}}(I)=\sum_{A\in\mathcal{A}}c_AI_A,$$
\item {\bf Uniswap V3 market makers} ~\cite[]{uniswap} defined for the case of two assets $$f_{\text{V3}}(I_A,I_B)=\sqrt{(I_A+\alpha)(I_B+\beta)},$$
for $\alpha,\beta\geq 0.$
\item {\bf Curve market makers}~\cite[]{egorov2019} which are implicitly defined by
$$\alpha n^n\sum_{A\in\mathcal{A}}I_A+f_{curve}(I)=\alpha n^nf_{curve}(I)+\frac{f_{curve}(I)^{n+1}}{n^n\prod_{A\in\mathcal{A}}I_A},$$
where $n:=|\mathcal{A}|$, $\alpha>0$ is a constant, and $f_{curve}(I)$ corresponds to the unique positive real solution of the above equation for inventory $I\in\mathbb{R}^{\mathcal{A}}_+$ with $I>0$.\footnote{In the following we use the notation $I>0$ to denote that all inventories are positive, i.e. $I_A>0$ for each $A\in\mathcal{A}$.}
\item {\bf Logarithmic scoring rule (LMSR)} market makers~\cite[]{hanson2003,prediction_markets} $$f_{LMSR}(I)=-b\log\left(\sum_{A\in\mathcal{A}}e^{(c_A-I_A)/b}\right),$$
for coefficients $b\neq 0$  and $c_A\in\mathbb{R}$ for $A\in\mathcal{A}$. Note that as $b\to \infty$, the LMSR converges to the constant sum rule. It is more customary for prediction markets to describe AMMs through cost functions. Following the above construction, we let $I^0\in\mathbb{R}_+^{\mathcal{A}}$ be an arbitrary initial inventory. 
Then, the cost function is given by $$C_{LMSR}(r)=-f(I^0-r)=b\log\left(\sum_{A\in\mathcal{A}}e^{-(I_A^0-r_A-c_A)/b}\right)=b\log\left(\sum_{A\in\mathcal{A}}e^{(r_A-\tilde{c}_A)/b}\right)$$ for $\tilde{c}_A:=c_A-I^0_A$ which is the usual LMSR cost function. Alternatively, the LMSR can be defined, as in the original definition due to~\cite{hanson2003}, by a logarithmic scoring rule, $S_A(p)=\tilde{c}_A+b\log(p_A)$ for $A\in\mathcal{A}$.\footnote{\label{footnote} The last equivalence can be seen as follows:
dimension $A\in\mathcal{A}$ of the scoring rule is the net profit of the trade $r$ if  (Arrow-Debreu) security $A$ pays out at settlement. Thus, $S_A(p)=r_A-C(r)$. The marginal price (see the subsequent section) of $A$ at trade $r$ is $p_A=\frac{\partial C(r)}{\partial r_A}=\frac{e^{(r_A-\tilde{c}_A)/b}}{\sum_{A\in\mathcal{A}}e^{(r_A-\tilde{c}_A)/b}}$.
Therefore, for the logarithmic score $S_A(p)=\tilde{c}_A+b\log(\frac{e^{(r_A-\tilde{c}_A)/b}}{\sum_{A\in\mathcal{A}}e^{(r_A-\tilde{c}_A)/b}})=r_A-b\log(\sum_{A\in\mathcal{A}}e^{(r_A-\tilde{c}_A)/b})=r_A-C(r)$ as required.}
\end{itemize}
\begin{figure}
    \centering
\includegraphics[width=7cm]{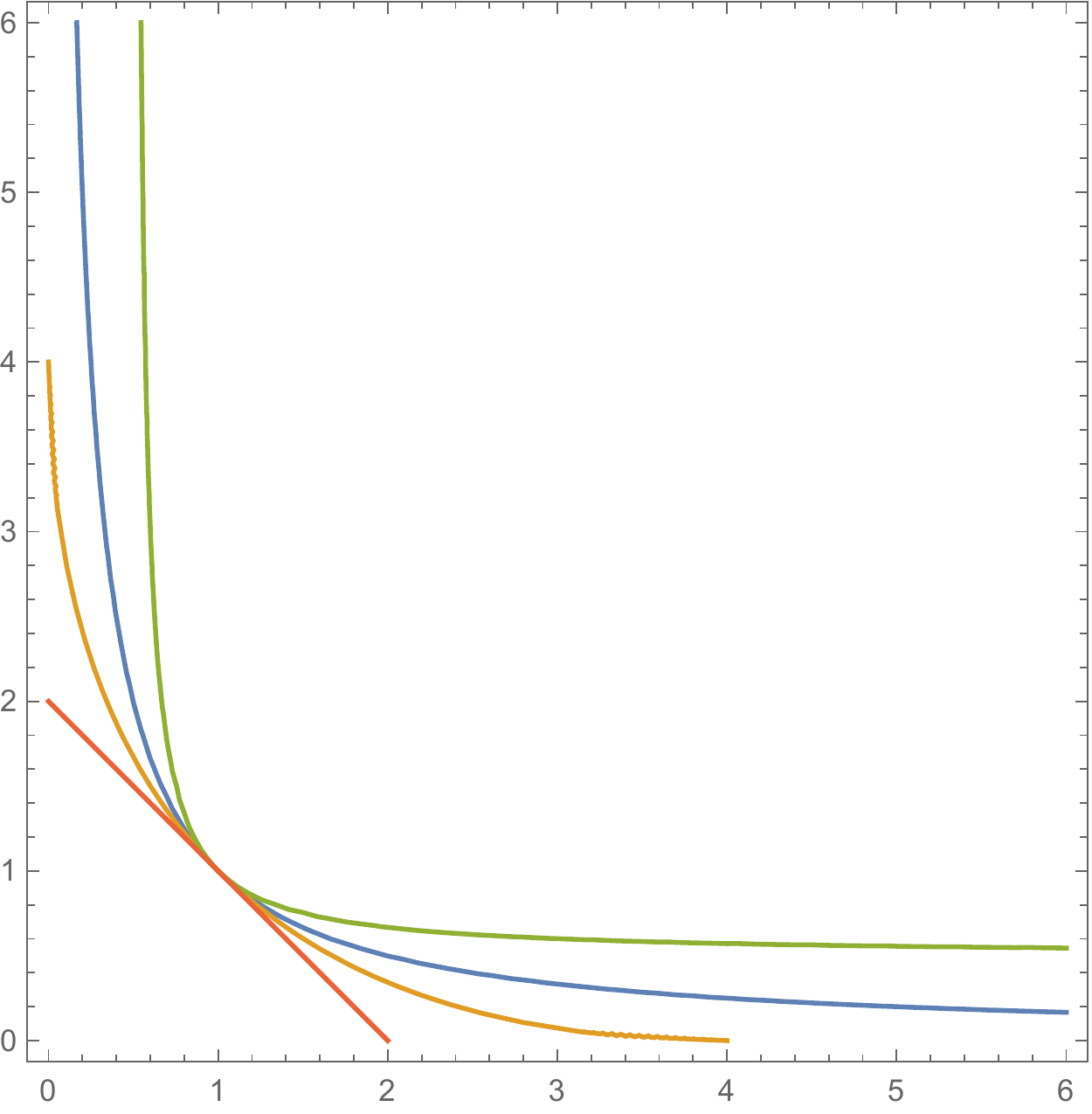}
    \caption{Different CEMMs:  elasticity of $1/2$ (green), i.e. elasticity of $1$ (blue) equivalent to constant product, the elasticity of $2$ (yellow), infinite elasticity (orange) equivalent to constant sum.}
    \label{fig:my_label2}
\end{figure}
Next, we define a class of AMMs that has been introduced under different names in DeFi~\cite[]{Niemerg2020,Othman2021,wu2022}.\footnote{The functional form is of course well known in economics, in particular in the theory of production to describe production processes where percentage changes in choices of production factors change proportionally to percentage changes in factor prices.} They are characterized by the property that, at the margin, the same percentage change in the exchange rate of two assets $A$ and $B$ always corresponds to the same percentage change in inventory ratio $I_A/I_B$ (See Remark~1). The constant sum and constant geometric mean market makers are special cases.
\begin{itemize}
    \item {\bf Constant inventory elasticity} AMMs ({\bf CEMMs}) parameterized by  $\gamma\in \mathbb{R}$,
$$f(I)=
c\left(\sum_{A\in\mathcal{A}}\alpha_AI_A^{\gamma}\right)^{1/\gamma},\quad\text{ for }\gamma\neq 0,$$
and$$
f(I)=c\prod_{A\in\mathcal{A}}I_A^{\alpha_A},\quad\text{ for }\gamma=0,
$$
for constants $c$ and $\alpha_A>0$ with $\sum_{A\in\mathcal{A}}\alpha_A=1.$
We call $\tfrac{1}{1-\gamma}$ the {\bf inventory elasticity} of the AMM. A CEMM has {\bf positive} constant inventory elasticity if $\gamma\leq 1$.  
\end{itemize}
\subsection{Basic properties of AMMs}
Next, we introduce several axioms on trading functions.
First, we consider a differentiability condition that allows us  to define marginal prices. Importantly, our characterization results will not assume differentiability, but it will be a consequence of the combination of other axioms that the characterized trading functions are smooth or equivalent to a smooth trading function.\newline

\noindent{\bf Existence of Marginal Prices}: The function
$f$
is differentiable everywhere.\newline

Note that all of the above examples satisfy this property.
Under this assumption, the {\bf marginal price} of asset$A\in\mathcal{A}$ measured in the numéraire is $\frac{\partial f}{\partial I_A}$. Thus, the {\bf marginal exchange rate} of assets $A$ and $B$ for $A,B\in\mathcal{A}$ is $$p_{A,B}(I)\equiv\frac{\tfrac{\partial f}{\partial I_A}}{\frac{\partial f}{\partial I_B}}.$$ This means that if the current inventory of the AMM is $I$, then swapping asset $A$ for asset $B$ has, at the margin, an exchange rate of $p_{A,B}(I).$

\begin{remark}\label{invelasticity}
For the CEMMs and the LMSRs defined above, there is a proportional relation between the inventory ratio $I_B/I_A$ resp the difference in inventory $I_B-I_A$ and changes in the marginal exchange rate between $A$ and $B$: 
For the CPMM a one percent increase in the inventory ratio leads to a one percent increase in the  marginal exchange rate. More generally, for a CEMM with parameter $\gamma$, a one percent increase in the inventory ratio leads to an increase of $1-\gamma$ percent in the marginal exchange rate. Related to this, the CPMM is characterized by the feature that the value of the inventory of asset type $A$ measured in the numéraire $\frac{\partial f }{\partial I_A}I_A$ is always equal to the value of the inventory of asset type $B$ measured in the numéraire $\frac{\partial f }{\partial I_B}I_B$.

For the LMSR, an increase of $I_B-I_A$ by one unit leads to a $b$ percent increase in the  marginal exchange rate. 
\end{remark}

%

%


%

It is a desirable property of an AMM to have an increasing marginal cost for swapping assets as the trade size increases. In other words, liquidity curves should be convex. This is related to the notion of (im)permanent loss for liquidity providers: The (im)permanent loss of a liquidity provider is given by the market exposure of the liquidity position: If the exchange rate of the two assets moves a lot, holding and providing liquidity performs worse than keeping the equivalent amount of assets A and B without providing liquidity. If the marginal cost for the appreciating asset is increasing then inducing a price change that induces (im)permanent loss is more expensive for the market. \newline

\noindent{\bf Aversion to permanent loss}: For each $J\in\mathbb{R}_+^{\mathcal{A}}$,  the set $\{I:f(I)\geq f(J)\}$ is convex.\newline

If the trading function $f$ satisfies the existence of marginal prices, then aversion to permanent loss for two assets is equivalent to assuming that the relative marginal price is decreasing as a function of $I_A$ along each liquidity curve, i.e. $p_{A,B}(I_A,I_B(I_A),I_{-A,B})$ is decreasing in $I_A$. All examples of trading functions considered above, except for CEMMs with negative elasticity, satisfy aversion to permanent loss. Aversion to permanent loss is assumed as a basic assumption in large parts of the literature on CFMMs, see~\citet{CFMM,angeris2021,Angeris2022abs,capponi2021}. Similarly, prediction market AMMs are usually assumed to have convex cost functions~\cite[]{Abernethy2011} which implies (by change of sign) that the trading function is concave and hence has convex liquidity curves.

Next, we consider the property that the AMM is always able to make the market no matter how large the trade size is. Geometrically this means that liquidity curves do not intersect the axes. Recall that we use $I>0$ to indicate that all dimensions of $I\in\mathbb{R}_+^{\mathcal{A}}$ are strictly positive and define:\newline

\noindent{\bf Sufficient Funds}: For each $I,J\in\mathbb{R}_+^{\mathcal{A}}$ with  $f(I)=f(J)$
we have $$I>0\Rightarrow J>0.$$ 
\begin{remark}
If the sufficient funds condition is violated, then we have concentrated liquidity, which means that the AMM accepts trades only up to a bounded size. This sometimes is desirable, for example for the purpose of risk management for liquidity providers. The Uniswap V3 rules are an example of a trading function with concentrated liquidity and therefore of a violation of the sufficient funds' condition.
CEMMs satisfy it for positive inventory elasticity smaller or equal $1$ and violate it for negative inventory elasticity and for inventory elasticity of more than $1.$
\end{remark}
The next requirement is that liquidity curves should be invariant under the scaling of the inventory. The interpretation is straightforward: If liquidity is added or removed from the AMM proportionally to the current inventory levels for each asset, then the exchange rate of the assetshould not change.\newline

\noindent{\bf Scale invariance}:
For each $I,J\in \mathbb{R}^{\mathcal{A}}_{+}$ and $\lambda>0$
$$f(I)=f(J)\Rightarrow f(\lambda I)=f(\lambda J).$$\newline
Invariance under scaling of the inventory has been assumed or discussed in previous work, usually in the stronger form of (positive) homogeneity~\citep{capponi2021} or (positive) linear homogeneity~\citep{Othman2021,Angeris2022abs}.
In practice, liquidity is tokenized so that it becomes fungible and tradable. For this to be possible we need that the value of the  inventory (that is tokenized) scales proportionally with the inventory.\newline

\noindent{\bf Homogeneity in Liquidity}:
For each $I\in \mathbb{R}_{+}^{\mathcal{A}}$ and $\lambda>0$ 
$$f(\lambda I)= \lambda f(I).$$

It is straightforward to see that each scale-invariant trading function is equivalent to a (linearly) homogeneous trading function. So we can obtain our subsequent results alternatively for equivalence classes of scale invariant trading functions or for homogeneous trading functions (analogously if we would require positive homogeneity of degree $k$ instead we would obtain the same classes of functions re-scaled by a monotonic transformation with the $k$-th power function).
\begin{remark}
From the above examples, the Uniswap V3 and Hanson's rule violate scale invariance whereas the other rules satisfy it.
This means that liquidity can be made fungible for all other rules.
\end{remark}

An alternative invariance property that is usually required for AMMs in prediction markets is that liquidity curves should be invariant under translations. The interpretation is straightforward in that context: for a prediction market a risk-less portfolio consisting of the same amount of each of the assets has the same payoff measured in the numéraire in each state of the world. Therefore buying it should always cost the same.\newline

\noindent{\bf Translation invariance}:
For each $I,J\in \mathbb{R}^{\mathcal{A}}_{+}$ and $\alpha \in\mathbb{R}$ denoting the all $1$ vector by $\mathbbm{1}$,
$$f(I)=f(J)\Rightarrow f(I+\alpha\mathbbm{1})=f( J+\alpha\mathbbm{1}).$$

Translation invariance is a standard assumption in the literature on AMMs for prediction markets, often in  a stronger version that describes the scaling properties of an individual trading function rather than equivalence classes of trading functions under monotone transformations:\newline

\noindent{\bf One invariance}:
For each $I\in \mathbb{R}^{\mathcal{A}}_{+}$ and $\alpha \in\mathbb{R}$,
$$f(I+\alpha\mathbbm{1})=f(I)+\alpha\mathbbm{1}.$$
An important consequence of translation invariance is that the AMM can equivalently descibed by a scoring rule.

Finally, we consider a symmetry axiom that requires that in the absence of additional information on the nature of the assets, the market should be made in the same way if the names of the assets are flipped:\newline

\noindent{\bf Symmetry}: For each  permutation $\pi:\mathcal{A}\rightarrow\mathcal{A}$  and $I\in\mathbb{R}_+^{\mathcal{A}}$ we have 
$$f(I)=f((I_{\pi(A)})_{A\in\mathcal{A}}).$$
All above examples (for the CEMMs and LMSRs coefficients need to be chosen the same for all assets, $c_A=c_B$ for $A,B\in\mathcal{A}$) satisfy symmetry.
\section{Characterization for more than two assets}\label{multiple}
In this section, we provide several characterizations of AMMs for more than two assets. In Section~\ref{results}, we consider extensions to the case of exactly two assets where it is necessary to introduce a different axiom.

Many popular AMMs used in practice, can be equivalently described with a separable liquidity curve of the form:
$$\sum_{A\in\mathcal{A}}\phi_A(I_A)=k.$$
Also, CPMMs, or more generally weighted geometric means, belong to this class because the curve can be brought to the above form by using a logarithmic transformation.\footnote{To cover the case of $0$ inventories, we would allow for negatively infinite values of $f$ on the boundary in that case.} Similarly, the LMSRs belong to this class because the curve can be brough to the above form by using an exponential transformation. Separability turns out to be a consequence of an independence property that can be a desirable feature of an AMM: we call the marginal exchange rate between two assets \emph{independent} from other assets if the marginal exchange rate is a function of the inventories of these two assets alone and does not depend on the inventory of a third asset.
\begin{example}\label{ex:non-independent}
Consider the following non-separable trading function for three assets:
$$f(I)=I_AI_B+I_AI_C+I_BI_C.$$
The marginal exchange rate between $A$ and $B$ is given by
$$p_{A,B}(I)=\frac{I_B+I_C}{I_A+I_C},$$ which depends on the inventory of asset $C$. 
\end{example}
Independence occurs naturally if different asset pairs are traded in autonomous AMMs that do not directly influence each other, which is a common design choice of DeFi protocols. Independence can also be desirable as a robustness property: Price manipulation, e.g. for the purpose of a front and back-running attack, is a frequent concern in decentralized finance. Independence eliminates the possibility of certain types of such attacks; for example, it is not possible to manipulate an exchange rate for an asset pair by providing liquidity for a different asset or by trading another asset pair not involved in the original pair. Generally, independence is a desirable property, unless the traded assets are highly correlated in value or are different representations of the same underlying type of asset. This would for example be the case for stablecoin pairs. In the latter case, it is often desirable that each trade affects all assets involved in the AMM. 
 
Independence can more generally be formulated without the assumption of the existence of marginal prices which allowed us to refer to marginal exchange rates. In the more general formulation, independence requires that the terms of trade for trading a subset of assets $\mathcal{B}\subseteq\mathcal{A}$ does not depend on the inventory levels for the other assets $\mathcal{A}\setminus \mathcal{B}$.\newline

\noindent{\bf Independence}:
For each subset of assets $\mathcal{B}\subseteq\mathcal{A}$ and inventories  $I,J$ we have
$$f(I_{\mathcal{B}},J_{-\mathcal{B}})= f(J)\Leftrightarrow f(I)= f(J_{\mathcal{B}},I_{-\mathcal{B}}).$$
\begin{remark}
It is also instructive to consider the independence axiom formulated in terms of cost functions: \newline

\noindent{\bf Independence (Cost Function Version)}:
For each subset of assets $\mathcal{B}\subseteq\mathcal{A}$ and trades  $r^B,q^B\in\mathbb{R}^{\mathcal{A}}$ with $r^B|_{-\mathcal{B}}=q^B|_{-\mathcal{B}}=0$ and $r^{-B},q^{-B}\in\mathbb{R}^{\mathcal{A}}$ with $r^{-B}|_{\mathcal{B}}=q^{-B}|_{\mathcal{B}}=0$  we have
$$C(r^{- B}+r^B)-C(r^{-B})=C(r^{-B}+q^B )-C(r^{- B})\Leftrightarrow C(q^{-B}+r^B)-C(q^{-B})=C(q^{-B}+q^B )-C(q^{-B}).$$
Thus independence requires that whether the cost of two trades $r^B,q^B$ involving only assets in $\mathcal{B}$ is the same does not depend on previous trades  that did not involve assets in $\mathcal{B}$.
\end{remark}
The independence axiom only has a bite for more than two assets. In that case, independence and monotonicity allow us to obtain an equivalent separable representation (see~\citet{krantz1971}, Section 6.11) of the trading function.
The possibility of such representation crucially depends on the assumption of more than two assets. In Section~\ref{results}, we show that we need an additional axiom, Liquidity Additivity, which we do not need for the case of more than two assets.
%
\subsection{Characterization under scale invariance}
Combining separability with scale invariance imposes additional structure on the functions $\{\phi_A\}_{A\in \mathcal{A}}$. With scale invariance, we can show that the functions $\{\phi_A\}_{A\in \mathcal{A}}$ have to be of the power form ($\phi_A=c_AI_A^\gamma$) or logarithmic form. In particular, the combination of scale invariance and separability implies the existence of marginal prices everywhere, i.e.~smoothness of the liquidity curve is a consequence of scale invariance and independence. This is the content of the following theorem:
\begin{theorem}\label{main_result}
A trading function for $|\mathcal{A}|>2$ assets satisfies independence and scale invariance if and only if it is equivalent to a CEMM. 
\end{theorem}
\begin{proof}
One readily checks that CEMMs satisfy both axioms.

Continuity, monotonicity, and independence imply by a classical result (see~\citet{krantz1971} Section 6.11) that $f$ is equivalent to a trading function $\tilde{f}$ satisfying the separability condition: there are one-dimensional, continuous, increasing functions $\{\phi_A\}_{A \in \mathcal{A}}$ such that
$$\tilde{f}(I)=\sum_{A\in\mathcal{A}}\phi_A(I_A).$$
The major part of the proof consists in showing that the functions $\{\phi_A\}_{A \in \mathcal{A}}$ are monomials or logarithms, up to addition by a constant.

We choose an asset $A \in \mathcal{A}$ and we fix arbitrary inventory levels $J$. For $t > 0$ we define $\alpha(t)$ and $\beta(t)$ implicitly by
$$\tilde{f}(e^{\alpha(t)} J) = \tilde{f}(e^t J_A, J_{-A}), \qquad \tilde{f}(e^{\beta(t)} J) = \tilde{f}(J_A, e^t J_{-A}).$$
Recall that $\tilde{f}$ is continuous and, by monotonicity,
$$\tilde{f}(J) < \tilde{f}(e^t J_A, J_{-A}) < \tilde{f}(e^t J) .$$
Thus, $\alpha(t)$ is well-defined and $0 < \alpha(t) < t$. Similarly, $\beta(t)$ is well-defined and $0 < \beta(t) < t$.

By scale invariance, we find that for every real $s$,
$$\tilde{f}(e^{s + \alpha(t)} J) = \tilde{f}(e^{s + t} J_A, e^s J_{-A}), \qquad \tilde{f}(e^{s + \beta(t)} J) = \tilde{f}(e^s J_A, e^{s + t} J_{-A}),$$
and separability implies that
$$\tilde{f}(e^{s + t} J_A, e^s J_{-A}) + \tilde{f}(e^s J_A, e^{s + t} J_{-A}) = \tilde{f}(e^s J) + \tilde{f}(e^{s + t} J).$$
It follows that
$$\tilde{f}(e^{s + \alpha(t)} J) + \tilde{f}(e^{s + \beta(t)} J) = \tilde{f}(e^s J) + \tilde{f}(e^{s + t} J).$$
In other words, if we define
$$\phi(t) = \tilde{f}(e^t J),$$
then $\phi$ is increasing and for every real $s$ and every $t > 0$,
\begin{equation}\label{eq:difference}\phi(s + \alpha(t)) - \phi(s) = \phi(s + t) - \phi(s + \beta(t)),\end{equation}
where $0 < \alpha(t) < t$ and $0 < \beta(t) < t$. We claim that the above properties imply that, up to addition by a constant, $\phi$ is an exponential function or a linear function.

Suppose first that $\phi$ has continuous second-order derivatives, and $\phi'$ is everywhere positive.
By a direct calculation and continuity of $\phi''$,
\begin{align*}& \lim_{t \to 0^+} \frac{ \phi(s + \alpha(t) + \beta(t)) - \phi(s + \beta(t))-( \phi(s + \alpha(t))-\phi(s)) }{\alpha(t) \beta(t)}
\\ \qquad & = \lim_{t \to 0^+} \int_0^1 \frac{\phi'(s + \alpha(t) + \beta(t) y)-\phi'(s  + \beta(t) y)}{\alpha(t)} dy
\\ \qquad & = \lim_{t \to 0^+} \int_0^1 \int_0^1 \phi''(s + \alpha(t) x + \beta(t) y) dx dy = \phi''(s) .\end{align*}
By combining the previous identity with equation~\eqref{eq:difference}, we obtain
$$\phi''(s) =\lim_{t \to 0^+} \frac{\phi(s + \alpha(t) + \beta(t)) - \phi(s + t)}{\alpha(t) \beta(t)}.$$
However, since $\phi'$ is positive and continuous, we have
\begin{align*}& \lim_{t \to 0^+} \frac{\phi(s + \alpha(t) + \beta(t)) - \phi(s + t)}{\alpha(t) \beta(t)} 
\\ & \qquad = \lim_{t \to 0^+} \frac{\phi(s + \alpha(t) + \beta(t)) - \phi(s + t)}{\alpha(t)+\beta(t)-t}\frac{\alpha(t) + \beta(t) - t}{\alpha(t) \beta(t)}
\\ & \qquad = \lim_{t \to 0^+} \int_0^1 \phi'(s + t + (\alpha(t) + \beta(t) - t) x) dx \times \frac{\alpha(t) + \beta(t) - t}{\alpha(t) \beta(t)} \\ & \qquad =  \phi'(s)\lim_{t \to 0^+} \frac{\alpha(t) + \beta(t) - t}{\alpha(t) \beta(t)} \, .\end{align*}
Thus, the limit $\gamma:=\lim_{t \to 0^+} \frac{\alpha(t) + \beta(t) - t}{\alpha(t) \beta(t)}$ exists and for  every real $s$ we have
$$\phi''(s) = \gamma \phi'(s).$$ By solving this differential equation, we conclude that for some real constants $c, d$,
\begin{equation}\label{eq:phi:exponential}\phi(s) = \begin{cases} c e^{\gamma s} + d & \text{if $\gamma \ne 0$,} \\ c s + d & \text{if $\gamma = 0$.} \end{cases}\end{equation}

In the general case, when $\phi$ is not assumed to be smooth, we use an approximation argument. Let $\eta_n$ be an approximation to the identity: $\eta_n(x) = n \eta(n x)$, where $\eta$ is a non-negative smooth compactly supported function, and let $\phi_n$ be the convolution of $\phi$ and $\eta_n$:
$$\phi_n(t) = \int_{-\infty}^\infty \phi(s) \eta_n(t - s) ds .$$
Then $\phi_n$ is smooth, and since $\phi$ is increasing, $\phi_n'$ is everywhere positive. Furthermore, it is straightforward to check that $\phi_n$ satisfies~\eqref{eq:difference}. By the argument given above under the assumption that $\phi$ is smooth, there is a real number $\gamma$ such that we have
$$\phi_n(s) = \begin{cases} c_n e^{\gamma s} + d_n & \text{if $\gamma \ne 0$,} \\ c_n s + d_n & \text{if $\gamma = 0$,} \end{cases}$$
where $c_n, d_n$ are appropriate constants. Since $\phi_n$ converges pointwise to $\phi$, we conclude that~\eqref{eq:phi:exponential} also holds in the general case, and our claim is proved.

Finally, we consider the extension of the ray through the origin and $J$. Recall that $\phi(t) = \tilde{f}(e^t J)$, and that $\alpha(t)$ was defined in such a way that
$$\tilde{f}(e^{s + t} J_A, e^s J_{-A}) = \tilde{f}(e^{s + \alpha(t)} J) = \phi(s + \alpha(t)) .$$
 By separability,
$$\phi_A(e^{s+t}J_A) - \phi_A(e^s J_A) = \tilde{f}(e^{s + t} J_A, e^s J_{-A}) - \tilde{f}(e^s J) = \phi(s + \alpha(t)) - \phi(s) .$$
Let $I_A>0$ and let $s>0$ such that $I_A=e^sJ_A$. We already know that $\phi'$ is continuous and everywhere positive and that $0<\alpha(t)<t$. Thus,
\begin{align*} \lim_{t \to 0^+} \frac{\phi_A(e^{t} I_A) - \phi_A(I_A)}{t}  & = \lim_{t \to 0^+} \frac{\phi(s + \alpha(t)) - \phi(s)}{t} \\ & = \lim_{t \to 0^+} \frac{\alpha(t)}{t} \int_0^1 \phi'(s + \alpha(t) x) dx \\ & = \lim_{t \to 0^+} \frac{\alpha(t)}{t} \phi'(s) \, .\end{align*}
Formally, in the above identity we pass to the limit as $t \to 0^+$ along a subsequence such that $\frac{\alpha(t)}{t}$ converges to some limit $\alpha_0$. It follows that $\phi_A$ is differentiable, and
$$I_A \phi_A'(I_A) = \alpha_0 \phi'(s)=\alpha_0\phi'(\log(I_A/J_A)).$$
We conclude that
$$\phi_A(I_A) = \begin{cases} c_A I_A^\gamma + d_A & \text{if $\gamma \ne 0$,} \\ c_A \log I_A + d_A & \text{if $\gamma = 0$,} \end{cases}$$
where $c_A$ and $d_A$ are appropriate constants. Since $\phi_A$ is increasing, we have $c_A > 0$ when $\gamma \ge 0$, and $c_A < 0$ when $\gamma < 0$.

The asset $A$ was chosen arbitrarily, and therefore
$$\tilde{f}(I) = \begin{cases} \displaystyle d + \sum_{A \in \mathcal{A}} c_A I_A^\gamma & \text{if $\gamma \ne 0$,} \\ \displaystyle d + \sum_{A \in \mathcal{A}} c_A \log I_A & \text{if $\gamma = 0$,} \end{cases}$$
where $d = \sum_{A \in \mathcal{A}} d_A$.

Let $c = \sum_{A \in \mathcal{A}} c_A$. By applying the monotonic transformation
$$M(z) = \begin{cases} (z - d)^{1 / \gamma} & \text{when $\gamma > 0$,} \\ e^{(z - d) / c} & \text{when $\gamma = 0$,} \\ (d - z)^{1 / \gamma} & \text{when $\gamma < 0$,} \end{cases}$$
we see that $\tilde{f}$ and therefore $f$ is equivalent to a trading function with constant inventory elasticity with parameter $\gamma$ and coefficients $\alpha_A := |c_A|$ when $\gamma \ne 0$ or $\alpha_A := c_A / c$ when $\gamma = 0$.
\end{proof}
It is straightforward to see that if the trading function additionally satisfies aversion to permanent loss then the inventory elasticity is positive and that if it additionally satisfies sufficient funds, then the inventory elasticity is greater than $1$. Therefore we have the following corollary:
\begin{corollary}\label{cor:aversion}
    A trading function for $|\mathcal{A}|>2$ assets
    \begin{enumerate}
    \item satisfies independence, scale invariance, and aversion to permanent loss if and only if it is equivalent to a CEMM with positive inventory elasticity. 
    \item satisfies independence, scale invariance, and sufficient funds if and only if it is equivalent to a CEMM with inventory elasticity $0<1/(1-\gamma)\leq 1$.
    \end{enumerate}
\end{corollary}
\begin{proof}
For the first part note that the liquidity curves defined by constant inventory elasticity are convex if and only if $\gamma\leq1$.  Thus, by the aversion to permanent loss, the inventory elasticity $\frac{1}{1-\gamma}$ is positive.

The second part follows from the observation that the surface defined by $\sum_{A\in\mathcal{A}}I_A^{\gamma}=k$ for $k>0$ intersects with the axes for $\gamma>0$ but not for $\gamma<0$. For $\gamma=0$ one immediately sees that geometric averages satisfy the sufficient funds' condition.
\end{proof}
Under the convexity assumption of aversion to permanent loss, we can describe AMMs in a dual way by the usual duality theory for convex sets and their support functions. Thus we can also obtain a dual version of the first part of the above corollary that we describe in Appendix~\ref{dual}.

It is straightforward to see that if scale invariance is replaced by homogeneity in the above characterizations the AMM is not only equivalent to a CEMM, but it actually \emph{is} a CEMM. We get the following version of Corollary~\ref{cor:aversion}.
\begin{corollary}\label{cor:homogen}
A trading function for $|\mathcal{A}|>2$ assets
    \begin{enumerate}
    \item satisfies independence, homogeneity in liquidity, and aversion to permanent loss if and only if it is equivalent to a CEMM with positive inventory elasticity. 
    \item satisfies independence, homogeneity in liquidity, and sufficient funds if and only if it is equivalent to a CEMM with inventory elasticity $0<1/(1-\gamma)\leq 1$.
    \end{enumerate}
\end{corollary}
\begin{proof}
We have a homogeneous liquidity function $f$ and a strictly increasing function $M$ such that 
$$f(I)=M(\sum_{A\in\mathcal{A}}\alpha_AI_A^{\gamma})$$
with $0\neq\gamma\leq 1$
respectively $$f(I)=M(\prod_{A\in\mathcal{A}}I_A^{\alpha_A})$$
for constants $\alpha_A>0$ with $\sum_{A\in\mathcal{A}}\alpha_A=1$. In the first case, by homogeneity in liquidity, the monotone transformation $M$ is of the form $$M(z)=cz^{1/\gamma}$$ for a constant $c>0$. For $\gamma=0$,  by homogeneity in liquidity, the monotone transformation is of the form $M(z)=cz.$
Therefore the AMM defined by $f$ is of the constant elasticity form.
\end{proof}
The class of functions can be further narrowed down by adding the symmetry axiom which is natural in most applications. In the case of symmetry, the coefficients in the constant inventory formula are the same, i.e. $\alpha_A=\alpha_B$ for $A,B\in\mathcal{A}$. CFMMs in the class of scale-invariant, independent, symmetric, and sufficiently funded AMMs can be ranked (e.g. according to the Arrow-Pratt criterion) by the convexity of their induced liquidity curves. Convexity of liquidity curves has a natural interpretation: If a liquidity curve is more convex than another liquidity curve starting from a balanced inventory, the latter offers better terms of trade to the trader than the former and vice versa for liquidity providers.\footnote{The issue becomes more complex when liquidity provision is endogenous and fee revenues are considered. \citet{capponi2021} study the optimal choice of curvature among convex combinations of the constant sum and the constant product rule in an environment where liquidity providers, investors, and arbitrageurs interact. In their model, equilibrium payoffs for liquidity providers and investors are non-monotone in the curvature.} More formally, define a partial order on trading functions such that $f\succsim g$ if for each pair $A,B\in\mathcal{A}$, for each inventory level $I$ with $I_A=I_B$, $f$ offers more favorable terms of trade than $g$, i.e.~for $x>0$ and $y,y'>0$ such that $f(I_A+x,I_B-y,I_{-A,B})=f(I)$ and $g(I_A+x,I_B-y',I_{-A,B})=g(I)$ we have $y\geq y'$. Then we have the following theorem:
\begin{theorem}\label{opt:product}
Let $|\mathcal{A}|>2$. The constant product rule $f_{product}$ is trader optimal within the class of scale-invariant, independent, symmetric, and sufficiently funded AMMs, that is for each scale-invariant, independent, symmetric, and sufficiently funded $g$ we have $f_{product}\succsim g$ where $f_{product}\sim g$ if and only if $g$ equivalent to $f_{product}$.
\end{theorem}
\begin{proof}
Let $J\in\mathbb{R}_+^{\mathcal{A}}$ with $J_A=J_B
=:\bar{I}$ for $A,B\in\mathcal{A}$. Let $A,B\in\mathcal{A}$ be two assets and consider the curve $I_B(I_A)$ implicitly defined by $f_{product}(I_{A,B},J_{-A,B})=f_{product}(J)$. The curve is given by the equation $$I_B(I_A)=\frac{\bar{I}^2}{I_A}.$$ Similarly, consider the curve $I_B(I_A)$ implicitly defined by $g(I_{A,B},J_{-A,B})=g(J)$ for an AMM $g$ with constant inventory elasticity $0<1/(1-\gamma)<1$. The curve is given by the equation $$I_B(I_A)=(2\bar{I}^{\gamma}-I_A^{\gamma})^{1/\gamma}.$$

Next note that for each $I_A>\bar{I}$ we have a larger inventory of asset $B$ after trading $I_A-\bar{I}$ units of asset $A$  for asset $B$  under $g$ than under $f_{product}$, since (as $\gamma<0$ and $I_A>\bar{I}$) $$(2\bar{I}^{\gamma}-I_A^{\gamma})^{1/\gamma}\geq \frac{\bar{I}^2}{I_A}.$$ Since the inventory of type $B$ is smaller under the product rule than under $g$, the trader has obtained more units of asset $B$  in exchange for $x=I_A-\bar{I}$ units of asset $A$ under the product rule than under $g$.
\end{proof}
\begin{remark}
We can alternatively consider the price impact a trader has on the AMM when trading against it.  Starting with an inventory $I$,  $f$ is locally more price-stable than $g$ if there is a $\bar{x}>0$ such that for each $\bar{x}>x>0$ and $y,y'>0$ with $f(I_A+x,I_B-y,I_{-A,B})=f(I)$ and $g(I_A+x,I_B-y',I_{-A,B})=g(I)$ the exchange rate under $f$ has moved less through the trade than under $g$:
$$p^f_{A,B}(
I_A+x,I_B-y,I_{-A,B})-p^f_{A,B}(I)\leq p^g_{A,B}(I_A+x,I_B-y',I_{-A,B})-p^g_{A,B}(I).$$ It is a straightforward calculation that starting from a balanced inventory $I$ with $I_A=I_B$ the product rule is locally more price-stable than any other CEMM with $\gamma<0$.
\end{remark}
\subsection{Characterization under Translation Invariance}
Prediction market AMMs usually satisfy a different invariance property than DeFi AMMs. So it is a natural question what happens if in the above characterizations scale invariance is replaced by translation invariance.  With translation invariance, we can show that the functions $\{\phi_A\}_{A\in \mathcal{A}}$ have to be of the exponential form ($\phi_A=e^{(c_A-I_A)/b}$) or linear form up to a constant. Therefore, instead of CEMMs we now get LMSR rules as the only independent rules that satisfy the desirable invariance property (interpreting the constant sum rule which also satisfies the two properties as the limit case of an LMSR as $b\rightarrow \infty$). Similarly as before, the combination of translation invariance and separability implies the existence of marginal prices everywhere, i.e.~smoothness of the liquidity curve is a consequence of translation invariance and independence. This is the content of the following theorem:
\begin{theorem}\label{secondmain}
A trading function for $|\mathcal{A}|>2$ assets satisfies independence and translation invariance if and only if it is equivalent to an LMSR market maker or a constant sum market maker. 
\end{theorem}
\begin{proof}
One readily checks that LMSR rules and the constant sum rules satisfy both axioms.

As in the proof of Theorem~\ref{main_result}, continuity, monotonicity, and independence imply that $f$ is equivalent to a trading function $\tilde{f}$ satisfying the separability condition: there are one-dimensional, continuous, increasing functions $\{\phi_A\}_{A \in \mathcal{A}}$ such that
$$\tilde{f}(I)=\sum_{A\in\mathcal{A}}\phi_A(I_A).$$
The major part of the proof consists in showing that the functions $\{\phi_A\}_{A \in \mathcal{A}}$ are exponentials or linear.

We choose an asset $A \in \mathcal{A}$ and fix arbitrary inventory levels $J$. For $t > 0$ we define $\alpha(t)$ and $\beta(t)$ implicitly by
$$\tilde{f}( J+\alpha(t)\mathbbm{1}) = \tilde{f}(J_A+t, J_{-A}), \qquad \tilde{f}(J+\beta(t)\mathbbm{1}) = \tilde{f}(J_A,  J_{-A}+t\mathbbm{1}_{-A}).$$
Recall that $\tilde{f}$ is continuous and, by monotonicity,
$$\tilde{f}(J) < \tilde{f}(J_A+t, J_{-A}) < \tilde{f}( J+t\mathbbm{1}) .$$
Thus, $\alpha(t)$ is well-defined and $0 < \alpha(t) < t$. Similarly, $\beta(t)$ is well-defined and $0 < \beta(t) < t$.

By translation invariance, we find that for every real $s$,
$$\tilde{f}( J+(s+\alpha(t))\mathbbm{1}) = \tilde{f}(J_A+s+t, J_{-A}+s\mathbbm{1}_{-A}), \qquad \tilde{f}(J+(s+\beta(t))\mathbbm{1}) = \tilde{f}(J_A+s,  J_{-A}+(s+t)\mathbbm{1}_{-A}).$$
and separability implies that
$$\tilde{f}(J_A+s+t, J_{-A}+s\mathbbm{1}_{-A}) + \tilde{f}(J_A+s, J_{-A}+(s+t)\mathbbm{1}_{-A}) 
 = \tilde{f}(J+s\mathbbm{1}) + \tilde{f}( J+(s+t)\mathbbm{1}).$$
It follows that
$$\tilde{f}( J+(s+\alpha(t))\mathbbm{1}) + \tilde{f}(J+(s+\beta(t))\mathbbm{1})  = \tilde{f}(J+s\mathbbm{1})  + \tilde{f}( J+(s+t)\mathbbm{1}).$$
In other words, if we define
$$\phi(t) = \tilde{f}(J+t\mathbbm{1}),$$
then $\phi$ is increasing and for every real $s$ and every $t > 0$,
$$\phi(s + \alpha(t)) - \phi(s) = \phi(s + t) - \phi(s + \beta(t)),$$
where $0 < \alpha(t) < t$ and $0 < \beta(t) < t$. Note that this functional equation is the same as equation~(\ref{eq:difference}). Thus, using the same argument as in the proof of Theorem~\ref{main_result} $\phi$ is smooth and there is a real number $\gamma$ such that we have
$$\phi(s) = \begin{cases} ce^{\gamma s} + d & \text{if $\gamma \ne 0$,} \\ c s + d & \text{if $\gamma = 0$,} \end{cases}$$
where $c, d$ are appropriate constants.

As in the proof of Theorem~\ref{main_result}, we consider the extension of the ray through the origin and $J$.
 By separability,
$$\phi_A(J_A+s+t) - \phi_A(J_A+s) = \tilde{f}(J_A+s+t, J_{-A}+s\mathbbm{1}_{-A}) - \tilde{f}(J+s\mathbbm{1}) = \phi(s + \alpha(t)) - \phi(s) .$$
For $I_A\geq0$, defining $s:=I_A-J_A$, dividing the previous equation by $t$ and passing to the limit as $t\to0$, we obtain by an analogous argument as in the proof of Theorem~\ref{main_result} that $\phi_A$ is differentiable and
$$\phi_A'(I_A) = \alpha_0\phi'(I_A-J_A).$$
We conclude that
$$\phi_A(I_A) =\phi_A(J_A)+\alpha_0\int_{0}^{I_A-J_A}\phi'(s)ds=\phi(0)+\alpha_0(\phi(I_A-J_A)-\phi(0))=\begin{cases} \tilde{c}_Ae^{\gamma I_A}+\tilde{d}_A & \text{if $\gamma \ne 0$,} \\ \tilde{c}_AI_A+\tilde{d}_A & \text{if $\gamma = 0$,} \end{cases}$$
where $\tilde{c}_A$ and $\tilde{d}_A$ are appropriate constants. The asset $A$ was chosen arbitrarily, and therefore
$$\tilde{f}(I) = \begin{cases} \displaystyle d-\sum_{A \in \mathcal{A}}e^{(\hat{c}_A-I_A)/b} & \text{if $0< b< \infty$,} \\ 
\displaystyle d +\sum_{A \in \mathcal{A}}e^{(\hat{c}_A-I_A)/b} & \text{if $b< 0$,} \\
\displaystyle d + \sum_{A \in \mathcal{A}} \tilde{c}_A I_A & \text{if $b= \infty$,} \end{cases}$$
where $d = \sum_{A \in \mathcal{A}} \tilde{d}_A$, $b=-1/\gamma$ and $\hat{c}_A=b\log(|\tilde{c}_A|)$. By applying the monotonic transformation
$$M(z) = \begin{cases} 

-b\log(d - z) & \text{when $0<b<\infty$,} \\ 
-b\log(z - d) & \text{when $b <0$,} \\ 
(z - d) / c & \text{when $b= \infty$,} \end{cases}$$
where $c:=\sum_{A\in\mathcal{A}}\tilde{c}_A$,
we see that $\tilde{f}$ and therefore $f$ is equivalent to a trading function that is either constant sum or a LMSR.
\end{proof}
It is straightforward to see that if the trading function additionally satisfies aversion to permanent loss 
then $b$ is non-negative:
\begin{corollary}\label{cor:trans}
 A trading function for $|\mathcal{A}|>2$ assets
satisfies independence, translation invariance, and aversion to permanent loss if and only if it is equivalent to an LMSR with $b>0$ or a constant sum market maker. 

\end{corollary}

It is also straightforward to see that if translation invariance is replaced by one invariance in the above characterizations the AMM is not only equivalent to an LMSR, but it actually \emph{is} an LMSR.
\begin{corollary}\label{cor:one}
 A trading function for $|\mathcal{A}|>2$ assets
satisfies independence, one invariance, and aversion to permanent loss if and only if it is an LMSR with $b>0$ or a constant sum market maker. 
\end{corollary}
\begin{proof}
By Theorem-\ref{secondmain}, there is a strictly increasing function $M$ such that 
$$f(I)=M(-b\log( \sum_{A \in \mathcal{A}}e^{(c_A-I_A)/b}))$$
for $b\neq \infty$
respectively $$f(I)=M(\sum_{A \in \mathcal{A}} c_A I_A)$$
for $b=\infty$ and coefficients with $\sum_{A}c_A=1$.
 In either case, by one-invariance of $f$, the monotone transformation $M$ is of the form $M(z)=z.$
\end{proof}
\section{Characterizations for two assets}\label{results}
In this section, we provide axiomatizations of different classes of CFMMs for the case of two assets. An important difference to the multi-dimensional case is that independence does not put any restriction on the trading function. On the other hand, separability is still with loss of generality.\footnote{Examples of non-separable and homogeneous CFMMs are the two-dimensional version of the curve family $f_{curve}$ defined above.} 
Therefore, we subsequently explore two directions. First, we study axiomatizations for the general, non-separable case, while otherwise maintaining the axioms from the previous section. This leads to a large class of "concave" CFMMs that can be described by concave increasing bijections on the unit interval. Second, we consider the separable case, where similar results as in the higher dimensional case can be obtained. To obtain separability, we introduce a new axiom "Liquidity Additivity" that is sufficient (and necessary) for the separability of trading functions.
\subsection{Characterization without Seperability}
The next characterization result states that the space of all symmetric, homogeneous, sufficiently funded CFMMs satisfying aversion to permanent loss can be described by concave increasing bijections on the unit interval. The bijection encodes how the AMM should react to imbalances in the liquidity ratio $I_B/I_A$, see Figure~2.
\begin{figure}[t]
    \centering
    \begin{minipage}{0.45\textwidth}
        \centering
        \includegraphics[width=0.9\textwidth]{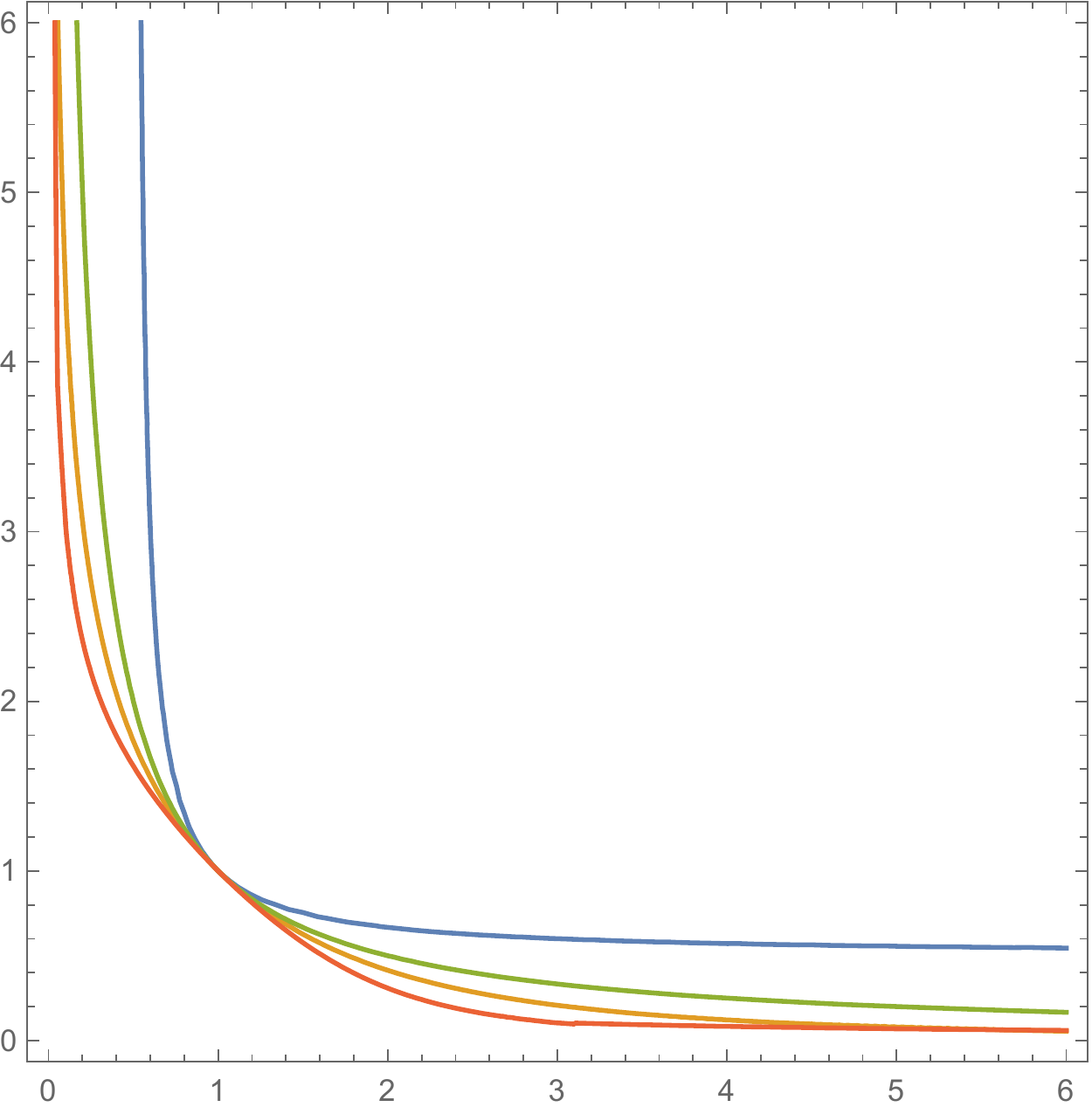} 

    \end{minipage}\hfill
    \begin{minipage}{0.45\textwidth}
        \centering
        \includegraphics[width=0.9\textwidth]{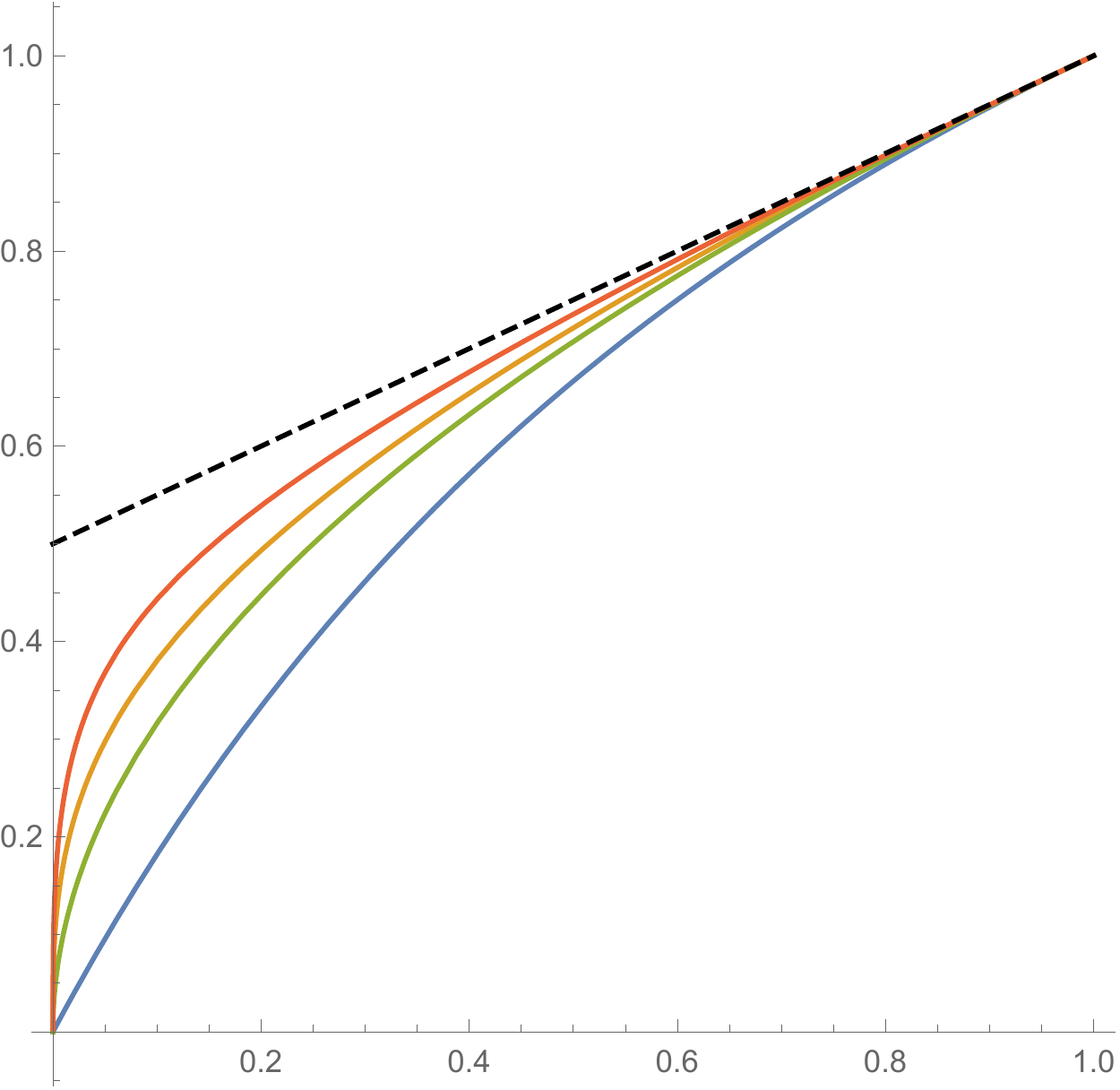} 
    \end{minipage}
      \caption{The left plot displays a liquidity curve and the right plot the function $g$ for a) the two-dimensional curve rule $f_{curve}$ with $\alpha=1$ (red) b) the rule defined by $f(I)=\sqrt[3]{I_A^2I_B+I_AI_B^2}$ (yellow), c) the constant product rule (green) d) the constant elasticity rule with $\gamma=-1$ (blue). The steeper the function in the beginning and the flatter the function in the end the more the exchange rate changes with trading if the inventory is imbalanced and the less it changes with trading if the  inventory is balanced. The dashed black line gives the upper envelope of permissible functions $g$.}
\end{figure}
\begin{theorem}\label{thm:bijection}
A trading function $f$ satisfies  homogeneity in liquidity, symmetry, sufficient funds, and aversion to permanent loss if and only if it is of the form  $$f(I)=\begin{cases}cI_Ag\left(\tfrac{I_B}{I_A}\right),\quad\text{ for }I_A\geq I_B\\
cI_Bg\left(\tfrac{I_A}{I_B}\right),\quad\text{ for }I_B\geq I_A\end{cases}$$ for a constant $c>0$ and an increasing and concave bijection $g:[0,1]\rightarrow[0,1],$ with $\lim_{z\to 1^-}g'(z)=\tfrac{1}{2}$ where the limit is along a sequence of points where $g$ is differentiable.
\end{theorem}
\begin{proof}
It is straightforward to see that an AMM that is represented in this way satisfies the axioms. For the opposite direction,  note that aversion to permanent loss means that $f$ is quasi-concave. By a theorem of~\citet{friedman1973}, monotonicity,  homogeneity, and quasi-concavity of $f$ imply that $f$ is a concave function. Next define $g(z):=f(1,z)/f(1,1)$ and $c:=f(1,1)$. By homogeneity we have $$cI_Ag(\frac{I_B}{I_A})=I_Af(1,\frac{I_B}{I_A})=f(I).$$
Concavity implies coordinate-wise concavity, and therefore concavity of $f$ implies that $g$ is concave. By symmetry, the representation also works for $I_B\geq I_A$.  By sufficient funds, for each $k>0$ the curve $I_B(I_A)$ implicitly defined by $f(I)=k$ satisfies $\lim_{I_A\rightarrow 0^-}I_B(I_A)=\infty$, where we take the limit along a sequence if inventories where the liquidity curve is differentiable. This is possible as convex real-valued functions are almost everywhere differentiable. Therefore we have $\lim_{z\rightarrow0^+}g(z)=\lim_{I_A\rightarrow 0^+}g(I_A/I_B(I_A)))=\lim_{I_A\rightarrow 0^+}\frac{k}{cI_B(I_A)}=0$ where the first limit is along a sequence of inventory ratios where $g$ is differentiable (this is possible as concave real-valued functions are almost everywhere differentiable). By continuity and homogeneity of $f$ we have $f(0,0)=0$. By monotonicity and continuity of $f$, $0=f(0,0)\leq f(1,0)/f(1,1)= g(0)\leq\lim_{z\rightarrow 0^+}g(z)=0$ and therefore $g(0)=0.$ By construction we have $g(1)=1.$ Since $g$ is increasing this implies that $g$ is a bijection of $[0,1]$. For any inventory $I$ where the liquidity curve is differentiable, with $z=I_B/I_A$ we have that
$$p_{A,B}(I)\equiv p_{A,B}(z)=\frac{g(z)}{g'(z)}-z.$$
Consider a sequence of inventory ratios approaching $1$ on which $g$ is differentiable (which exists by the same logic as above). By aversion to permanent loss and symmetry, the exchange rate along this sequence needs to approach $1$. Therefore $$1=\lim_{z\rightarrow1^-}p_{A,B}(z)=\lim_{z\to 1^-}\frac{g(z)}{g'(z)}-z=\frac{1}{\lim_{z\to 1^-}g'(z)}-1\Rightarrow\lim_{z\to 1^-}g'(z)=\tfrac{1}{2}.$$
\end{proof}

With the representation above the CPMM can be obtained through the function $g(z)=\sqrt{z}$, as
$$p_{A,B}(z)=\frac{g(z)}{g'(z)}-z$$
where $z\equiv I_B/I_A$ is the inventory ratio. The expression is linear in $z$ for the square root function.
\subsection{Separability Axioms}
Separability is more natural as an assumption for more than two assets since it is related to independence in that case (see Section~\ref{multiple}). For two assets, independence has no bite. To guarantee separability in the two-dimensional case, instead, we need that liquidity provision is additive in the following sense: imagine that a liquidity provider wants to provide liquidity in either of the two assets.\footnote{As argued above, in the decentralized exchange application, exchanging one token for the numéraire (LP tokens), i.e.~providing liquidity in one token, can always be emulated by first swapping some amount of that token for the other token and afterward providing liquidity for both tokens.} Suppose that starting from an inventory level $I$, the liquidity provider obtains the same amount of numéraire by either providing $x$ units of asset $A$ or by providing $y$ units of asset $B$. Suppose furthermore that starting from an inventory level $(I_A+x,I_B+y)$, the liquidity provider obtains the same amount of numéraire by either providing $\tilde{x}$ units of asset $A$ or by providing $\tilde{y}$ units of asset $B$.  Then starting from $I$, the liquidity provider should obtain the same amount of numéraire by providing $x+\tilde{x}$ units of asset $A$ as by providing $y+\tilde{y}$ units of asset $B$. This leads to the following condition which is often called the Thomsen condition~\cite{krantz1971} in decision theory:\newline

\noindent{\bf Liquidity Additivity}:
For inventory $I$ and quantities $x,y,\tilde{x},\tilde{y}$ such that
$$f(I_A+x,I_B)=f(I_A, I_B+y),\quad f(I_A+x+\tilde{x},I_B+y)=f(I_A+x,I_B+y+\tilde{y})$$
we have
$$f(I_A+x+\tilde{x},I_B)=f( I_A, I_B+y+\tilde{y}).$$\newline

\begin{figure}
    \centering
\includegraphics[width=9cm]{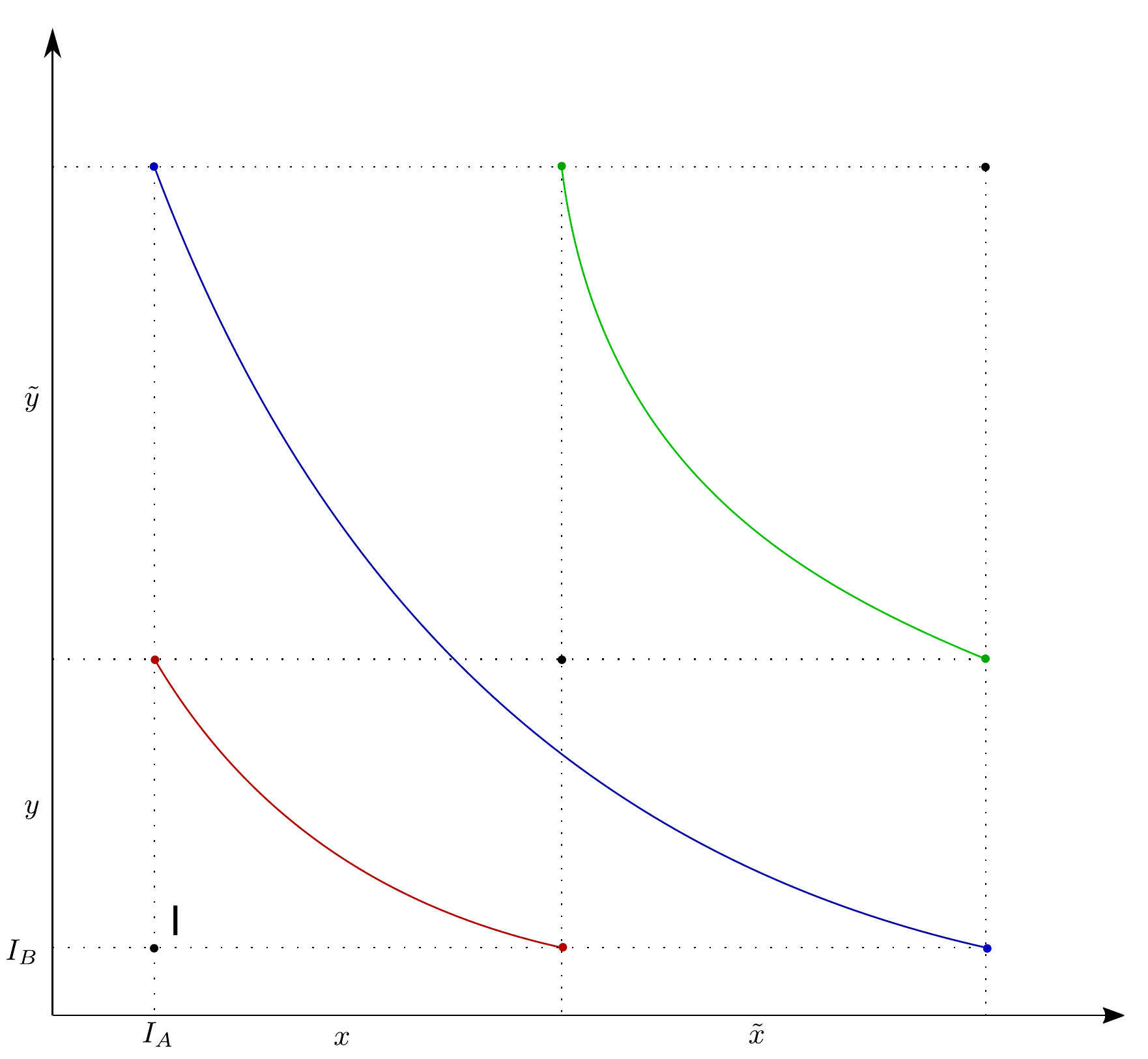}
    \caption{Three liquidity curves from the additivity axiom.}
    \label{fig:my_label}
\end{figure}
\begin{remark}
It is again instructive to consider the cost function version of the liquidity additivity axiom. \newline

\noindent{\bf Cost Seperability (Cost Function version of Liquidity Additivity)}:
For trades $q,r,s\in\mathbb{R}^2$ such that
$$ C(r+q)-C(r)=C(r+s)-C(r)$$
we have
$$C(r_A+q_A,r_B)-C(r)=C(r_A,r_B+s_B)-C(r)\Leftrightarrow C(r_A+s_A,r_B)-C(r)=C( r_A, r_B+q_B)-C(r).$$\newline
Thus, if two trades, $q,s$ cost the same, then the first dimension of one of the two trades cost the same as the second dimension of the other trade if and only if the other two dimensions cost the same.
\end{remark}
Replacing the independence axiom with the Liquidity Additivity axiom we obtain the following characterization for the two-dimensional case:

\begin{theorem}\label{elasticity}
A trading function satisfies Liquidity Additivity and scale invariance if and only if it is equivalent to a CEMM. 
\end{theorem}

\begin{proof}
One readily checks that CEMMs satisfy both axioms.

Monotonicity, continuity, and Liquidity Additivity imply by a classical result (see~\citet{krantz1971} Section 6.2) that $f$ is equivalent to a trading function $\tilde{f}$ for which there are one-dimensional, increasing functions $\phi_A,\phi_B$ such that $$\tilde{f}(I)=\phi_A(I_A)+\phi_B(I_B).$$
The construction of the additive representation proceeds by constructing a self-mapping of the plane that maps liquidity curves into lines with a negative slope, and lines parallel to the axes into lines parallel to the axes. The remainder of the proof follows the same logic as the proof of Theorem~1.
\end{proof}
We also obtain corresponding corollaries to Corollaries~\ref{cor:aversion} and~\ref{cor:homogen}. The proofs are completely analogous to the multi-dimensional case with independence replaced by Liquidity Additivity.
\begin{corollary}\label{cor:3}
A trading function for $|\mathcal{A}|=2$ assets
\begin{enumerate}
\item satisfies  Liquidity Additivity, scale invariance, and aversion to permanent loss if and only if it is equivalent to a CEMM with positive inventory elasticity.
\item satisfies Liquidity Additivity, scale invariance, and sufficient funds if and only if it is equivalent to an AMM with positive constant inventory elasticity $\tfrac{1}{1-\gamma}\leq1$.
\end{enumerate}
\end{corollary}
\begin{corollary}\label{cor:5}
A trading function for $|\mathcal{A}|=2$ assets
\begin{enumerate}
\item satisfies  Liquidity Additivity, homogeneity in liquidity, and aversion to permanent loss if and only if it is a CEMM with positive inventory elasticity.
\item satisfies Liquidity Additivity, homogeneity in liquidity, and sufficient funds if and only if it is a CEMM with elasticity $0<\tfrac{1}{1-\gamma}\leq1$.
\end{enumerate}
\end{corollary}
We also get an analogous result to Theorem~\ref{opt:product}. The proof is also analogous to the multi-dimensional case.
\begin{theorem}\label{opt:product2}
The constant product rule $f_{product}$ is trader optimal within the class of scale invariant, LP additive, symmetric, and sufficiently funded AMMs, that is for each scale invariant, LP additive, symmetric, and sufficiently funded $g$ we have $f_{product}\succsim g$ where $f_{product}\sim g$ if and only if $g$ equivalent to $f_{product}$.
\end{theorem}
Similarly, the result under translation invariance carries over to the two-dimensional case with independence replaced by Liquidity Additivity:
\begin{theorem}\label{elasticity2}
A trading function for $|\mathcal{A}|=2$ assets satisfies Liquidity Additivity and translation invariance if and only if it is equivalent to an LMSR market maker or a constant sum market maker. 
\end{theorem}
\begin{proof}
One readily checks that LMSRs satisfy both axioms.

Monotonicity, continuity, and Liquidity Additivity imply by a classical result (see~\citet{krantz1971} Section 6.2) that $f$ is equivalent to a trading function $\tilde{f}$ for which there are one-dimensional, increasing functions $\phi_A,\phi_B$ such that $$\tilde{f}(I)=\phi_A(I_A)+\phi_B(I_B).$$
The construction of the additive representation proceeds by constructing a self-mapping of the plane that maps liquidity curves into lines with a negative slope, and lines parallel to the axes into lines parallel to the axes. The remainder of the proof follows the same logic as the proof of Theorem~\ref{secondmain}.
\end{proof}
We also obtain corresponding corollaries to Corollaries~\ref{cor:trans} and~\ref{cor:one}. The proofs are completely analogous to the multi-dimensional case with independence replaced by cost seperability.
\begin{corollary}\label{cor:4}
 A trading function for $|\mathcal{A}|=2$ assets
satisfies Liquidity Additivity, translation invariance, and aversion to permanent loss if and only if it is equivalent to an LMSR with $b>0$ or a constant sum market maker. 

\end{corollary}
\begin{corollary}\label{cor:7}
 A trading function for $|\mathcal{A}|=2$ assets
satisfies Liquidity Additivity, one invariance, and aversion to permanent loss if and only if it is an LMSR with $b>0$ or a constant sum market maker. 
\end{corollary}

\section*{Acknowledgements}

{We thank Maxim Bichuch, Federico Echenique, Zach Feinstein, Tim Roughgarden, Ludvig Sinander and participants of Tokenomics 2022 for valuable comments.}

\bibliographystyle{ACM-Reference-Format}
\bibliography{sample}
\newpage
\begin{appendix}

\section{Logical Independence of Axioms}\label{LogicalIndependence}
We now show that the axioms in the different characterizations are all necessary to obtain the results.

For the logical independence of axioms in Theorem~\ref{main_result} we can use the following examples: First, the trading function from Example~\ref{ex:non-independent} is not independent but satisfies scale invariance. Note that this trading function is not a CEMM. The example can also be easily generalized to more than three assets. Second, consider the following example: $$f(I) := \sum_{A\in \mathcal{A}} I_A + \sqrt{I_A}.$$ The trading function is independent but violates scale invariance and is not a CEMM. 

For the logical independence of axioms in Corollary~\ref{cor:aversion} note that the above examples all satisfy aversion to permanent loss. The example $$f(I) = I_AI_BI_C(I_AI_B + I_BI_C + I_AI_C)$$ satisfies scale invariance and sufficient funds but violates independence and is not a CEMM. 
The sum of two CEMMs with parameters $\gamma_1<0$ and $\gamma_2<0$ such that $\gamma_1\neq\gamma_2$ is independent and sufficiently funded but violates scale invariance.
Finally, a trading function of constant but negative inventory elasticity, i.e.~with $\gamma>1$ is independent, satisfies scale invariance, but does not satisfy aversion to permanent loss, whereas e.g. the constant sum market maker is independent, satisfies scale invariance but violates sufficient funds.

For the logical independence of axioms in Corollary~\ref{cor:homogen}, we can use the same examples as above, up to ''rescaling" - by which we mean monotonically transforming the trading function from Example~\ref{ex:non-independent} into a homogeneous one by applying the square root function.

For the logical independence of axioms in Theorem~\ref{thm:bijection} we can use the following examples: The un-scaled product rule $f(I)=I_AI_B$ does not satisfy homogeneity but satisfies all other axioms.  The function does not have a representation as in the statement of the theorem. The function $f(I)=I_A^{\theta}I_B^{1-\theta}$, where $0<\theta<\frac{1}{2}$ satisfies all axioms except for symmetry and does not have a representation as in the statement of the theorem. The function $f(I)=I_A+I_B$ satisfies all axioms but sufficient funds. The function does not have a representation as in the statement of the theorem. The function$$f(I)=\begin{cases}\sqrt{I_A^2+I_B^2}, \quad &I_B/2\leq I_A\leq2I_B\\\sqrt{\tfrac{5}{2}I_AI_B}\quad &\text{else.}\end{cases}$$satisfies all axioms but aversion to permanent loss. The function does not have a representation as in the statement of the theorem.

For the logical independence of axioms in Theorem~\ref{elasticity} we can use the following examples: The trading function  $$f(I)=(I_AI_B^2+I_A^2I_B)^{1/3}$$ satisfies scale invariance, but not LP-additivity. The function does not have positive constant inventory elasticity. For Corollary~\ref{cor:3} note that the above examples also satisfy aversion to permanent loss and sufficient funds. The logical independence of the other axioms in the corollaries can be established by two-dimensional analogs of the examples that established logical independence for the multidimensional case. For Corollary~\ref{cor:5} note furthermore that the above example is also homogeneous. 

For the logical independence of axioms in Theorem~\ref{secondmain} and Corollary~\ref{cor:trans}, the trading function

\begin{equation*}
    f(I) = \frac{1}{n}\left(\sum_{A\in \mathcal{A}}I_A - \tfrac{1}{n-1}\sum_{A\neq B}\sqrt{1+(I_A-I_B)^2}\right). 
\end{equation*}
 satisfies translation invariance and aversion to impermanent loss but violates independence.  

CPMMs satisfy aversion to impermanent loss and independence, but violate translation invariance. The LMSR with $b<0$ satisfies translation invariance and independence but violates aversion to impermanent loss. 

All of the examples above that satisfy translation invariance also satisfy one invariance, therefore, we get the independence of axioms in the corollary~\ref{cor:one}. 

The examples described also work for the logical independence of axioms for 2-dimensional case, by simply choosing $n=2$ and noting that the relevant examples violate or satisfy Liquidity Additivity.

\section{Discontinuous Liquidity Curves}\label{Discontinuous}
Continuity of $f$ is generally necessary to obtain our results. However, Theorem~\ref{thm:bijection}, Corollaries~\ref{cor:homogen} and~\ref{cor:5}, can alternatively be obtained under the assumption that $f$ is continuous on the boundary, i.e.~$\lim_{I_A\to0}f(I)=f(0,I_{-A})$ for each $A\in\mathcal{A}$, since continuity on the boundary, aversion to permanent loss and homogeneity imply that $f$ is concave (see~\citet{friedman1973}) and therefore continuous.
The main characterization results fail to hold without continuity. In two dimensions the following example shows this: Choose $\phi_A$ and $\phi_B$ to be increasing functions that take on values in null sets $A\subset\mathbb{R}_{++}$ and $B\subset \mathbb{R}_{++}$ that have the property that
$$\alpha,\alpha'\in A,\beta,\beta'\in B: \alpha+\beta=\alpha'+\beta'\Rightarrow\alpha=\alpha'\text{ and }\beta=\beta'.$$
Such sets are easy to construct using decimal expansions: To define $\phi_A(I_A)$ consider a binary expansion of $I_A$ and let $\phi_A(I_A)$ be the real number defined by the binary expansion interpreted as decimal expansion (so, for example, we can choose $\phi_A(1)=1,$ $\phi_A(1.5)=1.1$ and $\phi_A(2)=10$, etc.). Define $\phi_B(I_B)=2\phi_A(I_A)$. 
Let $f(I)=\phi_A(I_A)+\phi_B(I_B)$. By construction, $f(I)=f(J)$ if and only if $I=J$. Thus, properties such as scale invariance and sufficient funds are trivially satisfied. The trading function also satisfies Liquidity Additivity.
For the case of more than two assets a similar trading function can be constructed that satisfies independence, and trivially, scale invariance and the sufficient funds' condition.
\section{Duality and Portfolio Value Functions}\label{dual}
There is a dual description of trading functions satisfying the aversion to permanent loss in terms of portfolio value functions, see \citet{angeris2021}. Next, we show that we can alternatively reformulate Corollary~\ref{cor:homogen} in a dual way by imposing an independence property on the portfolio value function.

A {\bf portfolio value function} induced by the trading function $f$ at a level $k>0$ is defined for each $p\in\mathbb{R}_+^{\mathcal{A}}$ by
$$V_k(p):=\inf\{p\cdot I:f(I)\geq k\}.$$
The interpretation is as follows: if an arbitrageur has access to an infinitely liquid external market where prices measured in some numéraire are currently given by $p$, then starting from any inventory level $J$ with $f(J)=k$, the arbitrageur can maximally extract $p\cdot J-V_k(p)$ units of the numéraire good. Equivalently, the value of the inventory of the AMM measured in the numéraire on the external market and after interaction with the arbitrageur is $V_k(p)$.

Similarly to the independence axiom for trading functions, we can formulate an independence axiom for portfolio value functions that states that the value of a portfolio is separable in the different assets.\newline

\noindent{\bf Independence for portfolio value functions}:
For each subset of assets $\mathcal{B}\subseteq\mathcal{A}$ and price vectors  $p,q$ we have
$$V_k(p_{\mathcal{B}},q_{-\mathcal{B}})= V_k(q)\Leftrightarrow V_k(p)= V_k(q_{\mathcal{B}},p_{-\mathcal{B}}).$$

We can reformulate the content of Corollary~\ref{cor:homogen} in a dual way, by characterizing AMMs with positive constant inventory elasticity by independence.
The result is a straightforward consequence of Corollary~\ref{cor:homogen}: 
\begin{corollary}[Dual version of Corollary~\ref{cor:homogen}, Part 1]\label{cor:portfolio}
Let $f$ be a trading function satisfying the aversion to permanent loss so that it can be equivalently described by a family of portfolio value function $\{V_k\}_{k\geq0}$. The following statements are equivalent:
\begin{enumerate}

\item The portfolio value functions $\{V_k\}_{k>0}$ are re-scaled versions of each other, $V_k(p)=kV_1(p)$ for each $k>0$, strictly increasing in prices and independent.
\item The trading function $f$ is of constant inventory elasticity with parameter $\gamma<1$.
\item The portfolio value functions $\{V_k\}_{k>0}$ are of the form
$$V_k(p)=c\left(\sum_{A\in\mathcal{A}}\alpha_Ap_A^{\tfrac{\gamma}{\gamma-1}}\right)^{\tfrac{\gamma-1}{\gamma}}k,$$

when $0\neq \gamma< 1$ resp.
$$V_k(p)=c\left(\prod_{A\in\mathcal{A}}p_A^{\alpha_A}\right)k,$$
when $\gamma=0$
for $c>0$ and $\alpha_A>0$ with $\sum_{A\in\mathcal{A}}\alpha_A=1$.
\end{enumerate}
\end{corollary}
\begin{proof}First note that, by definition, the portfolio value function is homogeneous in prices $V_k(\lambda p)=\lambda V_k(p)$. By the same mathematical argument as in the proof of Theorem~\ref{main_result} and Corollary~\ref{cor:homogen}, independence, strict monotonicity in prices, and homogeneity of $V_1$ imply that it is of the form
$$V_1(p)=\begin{cases}c\left(\sum_{A\in\mathcal{A}}p_A^{\rho}\right)^{1/\rho},\quad &\rho\neq 0,\\
c\prod_{A\in\mathcal{A}}p_A^{\alpha_A},\quad &\rho=0,
\end{cases}$$
for a parameter $\rho\in\mathbb{R}$. By definition, portfolio value functions are concave in prices and therefore $\rho\leq 1$. Since the trading function $f$ is  strictly increasing, we have strict inequality $\rho< 1$. Define $\gamma=\tfrac{\rho}{\rho-1}$. Then we get that $V_1$ is of the form in (3).
If the portfolio functions are a re-scaled version of each other, we also get that all $V_k$ are of the form in (3). The opposite direction that (3) implies (1) can be verified immediately. The equivalence between (2) and (3) follows from straightforward calculations.
\end{proof}
\begin{remark}
Note that the constant mean AMMs (the case $\gamma=1$) are excluded in the above characterization. Since exchange rates are constant for the mean, the portfolio value function is not strictly increasing in prices in that case. Furthermore, portfolio value functions that are linear in prices are also excluded because there is no strictly increasing trading function that induces this value function. In both cases, the dual function of the linear function would be given by a weighted minimum. 
\end{remark}
\end{appendix}

\end{document}